\documentclass[twocolumn]{IEEEtran}

\usepackage{cite}
\usepackage{amsfonts,amsmath,amssymb,amsthm}
\usepackage{graphicx,color,epsfig,rotating}
\usepackage{latexsym}
\usepackage{subfigure}
\usepackage{cite}
\usepackage{epic,eepic}
\usepackage{algorithm}
\usepackage{algpseudocode}
\usepackage{float}
\usepackage{multirow}
\def\BibTeX{{\rm B\kern-.05em{\sc i\kern-.025em b}\kern-.08em    T\kern-.1667em\lower.7ex\hbox{E}\kern-.125emX}}
\usepackage{bbm}
\usepackage{bm}

\newtheorem{theorem}{Theorem}
\newtheorem{lemma}{Lemma}

\newtheorem{corollary}{Corollary}

\definecolor{caribbeangreen}{rgb}{0.0, 0.8, 0.6}

\usepackage[hidelinks]{hyperref}
\usepackage{orcidlink}

\newcommand{\EE}{\mathbb E} 
\newcommand{\PP}{\mathbb P} 

\newcommand{\sv}{{\mathbf{s}}}

\newcommand{\vv}{{\mathbf{v}}}

\newcommand{\cC}{{\mathcal C}}
\newcommand{\cD}{{\mathcal D}}

\newcommand{\cG}{{\mathcal G}}
\newcommand{\cH}{{\mathcal H}}

\newcommand{\cM}{{\mathcal M}}
\newcommand{\cN}{{\mathcal N}}

\newcommand{\cU}{{\mathcal U}}

\newcommand{\tP}{{\texttt P}}

\renewcommand{\d}{\operatorname{d\!}{}} 

\newcommand{\be}{\begin{equation}}
\newcommand{\ee}{\end{equation}}
\newcommand{\bea}{\begin{eqnarray}}
\newcommand{\eea}{\end{eqnarray}}
\newcommand{\bitem}{\begin{itemize}}
\newcommand{\eitem}{\end{itemize}}

\newcommand{\bdp}{\begin{displaymath}}
\newcommand{\edp}{\end{displaymath}}

\allowdisplaybreaks

\begin{document}

\title{Simultaneous Multi-Modal Covert Communications: Analysis and Optimization}

\author{ \IEEEauthorblockN{Justin H. Kong\raisebox{0.5ex}{\orcidlink{0000-0003-2856-7060}},~\emph{Senior Member, IEEE}, Terrence J. Moore\raisebox{0.5ex}{\orcidlink{0000-0003-3279-2965}},~\emph{Member, IEEE}, and \\  Fikadu T. Dagefu\raisebox{0.5ex}{\orcidlink{0000-0002-7532-5278}},~\emph{Senior Member, IEEE}  } \\

\thanks{The authors are with the U.S. Army Combat Capabilities Development Command (DEVCOM) Army Research Laboratory, Adelphi, MD 20783, USA (e-mail: justin.h.kong2.civ@army.mil; terrence.j.moore.civ@army.mil; fikadu.t.dagefu.civ@army.mil).
} } \maketitle



\begin{abstract}

This paper investigates the problem of covert communications in a heterogeneous wireless network where multiple communication modalities are used simultaneously. In this setup, a legitimate transmitter sends confidential data to its receiver by selecting multiple modalities with the goal of maximizing communication covertness against a passive adversary (Willie) while satisfying a transmission rate requirement. We analyze two distinct scenarios for a given observation time by Willie. The two scenarios are: (i) Willie knows the modalities selected by the friendly transmitter, and (ii) Willie is unaware of the selected modalities. We first derive the optimal detector for Willie that minimizes the detection error probability (DEP) in both cases. For the first scenario, we derive an exact expression for the DEP and provide a computationally efficient approximation. For the second scenario, we introduce the DEP expressions in the low-signal-to-noise ratio (SNR) regime at Willie. Building on this analysis, we propose a novel low-complexity modality set selection technique designed to maximize the DEP subject to a rate constraint. Numerical simulations validate the derived analytical expressions and demonstrate that the proposed modality set selection technique achieves near-optimal performance, outperforming benchmark schemes.

\end{abstract}


\IEEEpeerreviewmaketitle

\section{Introduction}

\IEEEPARstart{C}{overt} communication techniques that mask the very existence of transmissions from potential adversaries have been extensively investigated for security-sensitive networks where merely safeguarding the content of transmitted data is insufficient~\cite{Yan:19}. 
Moreover, heterogeneous networks leveraging various complementary communication modalities have been studied to enhance network security efficiently~\cite{Wu:18b}.
In this context, we explore covert communications in heterogeneous networks where multiple modalities are concurrently utilized, and analyze the impact of multi-modal selection on covertness performance.

\subsection{Related Work and Motivation}

Due to the open nature of the wireless channel, ensuring communication security against adversaries is a paramount concern in wireless communication networks~\cite{Liu:17,Wu:18b,Yan:19}. 
Conventional wiretap channel-based physical layer security approaches prioritize preventing the eavesdropper from deciphering the message content~\cite{Wyner:75}. 
However, in security-critical applications where the detection of transmission could trigger countermeasures, protecting only the data is often insufficient. 
This necessitates covert communication to hide the very existence of transmissions~\cite{Hero:03}, ensuring stealthy communication.

Foundational information-theoretic studies in~\cite{Bash:13, Che:13, Wang:16c, Arumugam:16} established the theoretical bounds of covert communications for an additive white Gaussian noise (AWGN) channel, a binary symmetric channel, a discrete memoryless channel, and a multiple access channel, respectively. 
Building on these foundations, recent works developed practical strategies to enable covert communications in intelligent reflecting surface (IRS)-assisted networks~\cite{Kong:21, Lv:22, Kong:24}, unmanned aerial vehicle (UAV)-enabled networks~\cite{Rao:22,Hou:23,Deng:24}, satellite communication (SATCOM) systems~\cite{Feng:24, Jia:25, Yu:25}, and integrated sensing and communications (ISAC) frameworks~\cite{Ma:23,Wang:24}. 

Recently, covert communications have been extended to heterogeneous networks that opportunistically select a modality from communication modalities having complementary physical layer characteristics —ranging from low-very high frequency (low-VHF) to microwave and millimeter-wave (mmWave) bands—to enhance communication covertness.
The authors in~\cite{Kong:22WCNC} analyzed the performance gain in covert communications achieved by dynamically switching between low-VHF and microwave modalities based on their distinct propagation and detection characteristics. 
Similarly, \cite{Zhang:24} explored mode selection to optimize covert throughput in hybrid microwave/mmWave UAV-enabled networks. 
Expanding into multi-hop routing scenarios, subsequent investigations addressed the joint covert routing and modality selection problem by developing a centralized optimization framework~\cite{Kong:24TIFS}, a decentralized reinforcement learning-based algorithm ~\cite{Kong:24CL}, and a joint resource allocation scheme for multi-flow networks~\cite{Kong:25}.

However, limiting transmission to a single modality may underutilize available spectral resources. 
Consequently, the improvement in communication covertness achieved by using multiple modalities simultaneously has been explored in~\cite{Aggarwal:24, Haque:25}. 
Considering two communication modalities, \cite{Aggarwal:24} analyzed an approximation of the detection error probability (DEP) at an adversary (Willie) employing an iterative detection threshold optimization.  
Furthermore, \cite{Haque:25} expanded this simultaneous two-modality transmission paradigm to the network layer, proposing a decentralized energy-efficient routing protocol to enable covert multi-hop communications.
In the parallel domain of multi-carrier systems, where a single frequency band is divided into multiple sub-carriers, \cite{Lin:23} examined delay-constrained covert communications using a lower bound of the DEP.

While the aforementioned works investigate the impact of resource optimization, covertness is also fundamentally enhanced when Willie suffers from uncertainty regarding the wireless environment. 
It is well established that when Willie lacks perfect knowledge of the environment, his detection capability is reduced, thereby improving communication covertness. 
Specifically, \cite{He:17} demonstrated that uncertainty regarding the background noise power at Willie prevents accurate power detection, enabling positive covert rates. 
Similarly, it was shown in~\cite{Shahzad:18} that the uncertainty induced by artificial noise from a full-duplex receiver further impedes Willie's ability to distinguish signal from noise. 
In addition, \cite{Hu:19} and~\cite{Shahzad:21} studied the impact of channel uncertainty, revealing that Willie's lack of perfect channel state information (CSI) significantly expands the covert transmission region. 
Furthermore, \cite{Wu:24} explored the ambiguity regarding the transmit prior probabilities, illustrating that Willie's uncertainty about Alice's transmission statistics creates additional opportunities for stealthy communication.

To fully exploit these environmental uncertainties and the distinct characteristics of heterogeneous networks, strategic modality set selection is essential for optimal exploitation of the unique complementary features of multiple communication modalities. 
In the context of single-modality systems with multiple subcarriers, optimization techniques for subcarrier selection were investigated to maximize covert rates~\cite{Che:22,Che:24,Lin:23}. 
The works in~\cite{Che:22} and~\cite{Che:24} analyzed single subcarrier selection scenarios where Willie is unaware of the selected frequency, relying on an approximation of the DEP.
In contrast, for scenarios where multiple subcarriers are selected simultaneously and are known to Willie, a resource allocation scheme based on a lower bound of the DEP was developed in~\cite{Lin:23}.

Despite these advances, existing research on simultaneous multi-modal (or multi-carrier) transmission relies heavily on approximations or a lower bound of the DEP~\cite{Aggarwal:24, Haque:25, Lin:23}, rather than an exact analytical derivation based on the optimal detector at Willie. 
Even in scenarios where Willie is fully aware of the selected modalities, an exact analysis of the DEP is not available in the literature.
More critically, the problem of modality uncertainty—where the transmitter simultaneously utilizes a subset of modalities without Willie knowing the specific set selected—has not been investigated. 
Consequently, the potential covertness gains achieved by adaptively selecting a set of modalities have not been investigated.

\subsection{Contributions and Organization}

In this paper, we investigate covert communications in heterogeneous networks where a legitimate transmitter (Alice) aims to reliably send confidential information to a legitimate receiver (Bob) while concealing the very existence of the transmission from Willie. 
Specifically, we focus on a simultaneous multi-modal transmission framework where Alice concurrently utilizes a subset of available modalities to fully exploit their complementary propagation characteristics. 
We analyze two distinct scenarios based on Willie's knowledge: (i) Willie knows the set of modalities selected by Alice, and (ii) Willie does not know the selected set of modalities. 
The main contributions of this paper are summarized as follows:

\begin{itemize}

\item We derive the optimal decision rule for Willie under the scenario where the set of active modalities is known to Willie. Based on this optimal detector, we provide a rigorous analytical derivation of the exact DEP. To facilitate practical analysis and optimization, we further develop a tractable closed-form approximation of the DEP by modeling the test statistic at Willie as a gamma distribution via a two-moment matching method.

\item We investigate the scenario where Willie is unaware of the specific subset of modalities employed by Alice, compelling him to monitor the spectrum of all available modalities. We identify the structure of the optimal detector for this case; however, as the exact DEP analysis is mathematically intractable due to the complex decision boundary, we derive an expression for the DEP in the low signal-to-noise ratio (SNR) regime at Willie—a condition inherent to covert communications. Additionally, we introduce a simplified approximation based on two-moment matching to enable efficient performance evaluation.

\item We formulate a modality set selection problem to maximize the DEP subject to a minimum rate requirement at Bob. To address the combinatorial complexity of the subset selection, we propose a low-complexity algorithm based on a novel metric that quantifies the achievable rate relative to the associated covertness cost. Additionally, we consider three distinct levels of CSI availability at Alice regarding the channels to Willie: (i) full instantaneous CSI, (ii) statistical CSI, and (iii) no CSI.

\item We provide extensive numerical simulation results to validate the accuracy of our derived analytical expressions and the tightness of the proposed approximations. We also evaluate the impact of simultaneous multi-modal selection on covertness performance and verify the efficiency of the proposed algorithm compared to benchmark schemes. Furthermore, we explore how communication covertness varies depending on Willie's knowledge of the active modalities and the available CSI at Alice.

\end{itemize}

The remainder of this paper is organized as follows. 
Section~\ref{sec:network_model} describes the network model. 
In Section~\ref{sec:analysis}, we derive the DEP for the scenario where Willie is fully aware of the selected set of modalities. 
Section~\ref{sec:unknown} extends this analysis to the case where Willie is unaware of the specific subset selected from the available set of modalities. 
Subsequently, Section~\ref{sec:optimization} proposes a novel low-complexity modality set selection algorithm to maximize the DEP while ensuring a minimum desired rate at Bob.
Numerical results are provided in Section~\ref{sec:results} to validate the theoretical analysis and verify the efficiency of the proposed technique. 
Finally, Section~\ref{sec:conclusion} concludes the paper.

Key notations used throughout this paper are summarized in Table~\ref{Table_notation}.
In this paper $x \sim \cC\cN(a,b)$ denotes that $x$ follows a complex Gaussian distribution with mean $a$ and variance $b$, $x \sim \chi^2(n)$ means that $x$ follows a chi-squared distribution with $n$ degrees of freedom, and $x \sim \cG(\alpha,\theta)$ indicates that $x$ follows a gamma distribution with shape parameter $\alpha$ and scale parameter $\theta$.

\begin{table}
\caption {List of Notations}  \vspace{-0mm}
\centering %
\begin{tabular}{|c|c|}
\hline
Notation & Definition \\
\hline
\hline
$M$ & Number of all available modalities     \\
\hline
$\cM$ & Set of all available modalities     \\
\hline
$\sv$ & Selected subset of $\cM$, i.e., $\sv \subseteq \cM$    \\
\hline
$M_{\sv}$ & Number of selected modalities, i.e., $M_{\sv} = |\sv|$    \\
\hline
$L$ & Number of channel uses in a block    \\
\hline
$P_m$ & Transmit power at Alice for modality $m$    \\
\hline
$\Omega_m$ & Allocated Bandwidth of modality $m$    \\
\hline
$N_0$ & Noise power spectral density   \\
\hline
$g_{\text{W},m}$ & Channel between Alice and Willie in modality $m$    \\
\hline
$g_{\text{B},m}$ & Channel between Alice and Bob in modality $m$    \\
\hline
$x_m[l]$ & Transmit data symbol in modality $m$ and channel use $l$     \\
\hline
$n_{\text{W},m}[l]$ & AWGN at Willie in modality $m$ and channel use $l$    \\
\hline
$n_{\text{B},m}[l]$ & AWGN at Bob in modality $m$ and channel use $l$    \\
\hline
$\rho_{\text{W},m}$ & SNR at Willie in modality $m$     \\
\hline
$\rho_{\text{B},m}$ & SNR at Bob in modality $m$     \\
\hline
$\tau$ & Acceptable outage probability at Bob   \\
\hline
$U_{m}$ & Rate at Bob in modality $m$   \\
\hline
$U_{\sv}$ & Sum-rate at Bob with the modalities in $\sv$    \\
\hline
$U_{\text{target}}$ & Target rate at Bob   \\
\hline
$\tP_{\text{DEP},\sv}$ & DEP at Willie with the modalities in $\sv$    \\
\hline
\end{tabular} \label{Table_notation}
\end{table}

\section{Network model} \label{sec:network_model}

In this paper, we examine heterogeneous networks where Alice sends confidential data to Bob using multiple communication modalities simultaneously, while reducing the probability of being detected by Willie as illustrated in Fig.~\ref{figure:model}.
All nodes are equipped with $M$ modalities, where each modality $m_i$ has a distinct operating center frequency $f_{m_i}$, $i=1,\dots,M$, resulting in unique channel propagation characteristics. 
We define $\cM = \{ m_1, \dots, m_M\}$ as the set of all available modalities.
Willie is a passive adversary attempting to detect communication by monitoring the frequency bands associated with all $M$ modalities.

We consider covert communications with finite blocklength~$L$ under a quasi-static fading model, where channels remain constant within a block but vary independently across blocks.
It is assumed that Alice decides whether to transmit data to Bob with equal \textit{a priori} probability, i.e., $\mathbb{P}(\mathcal{H}_0) = \mathbb{P}(\mathcal{H}_1) = 0.5$. 
Here, $\mathcal{H}_0$ stands for the null hypothesis representing the absence of transmission and $\mathcal{H}_1$ designates the alternative hypothesis where a transmission occurs.

\subsection{Detection at Willie}

Let $\sv \subseteq \cM$ denote the set of modalities selected by Alice for transmission to Bob. 
Then, the received signal at Willie on modality $m$ in channel use $l$ is given by
\begin{align} \label{eq:Received_signal_Willie}
	&\cH_0:  	y_{\text{W},m}[l] = n_{\text{W},m}[l], ~ \forall m \in \cM,  \\ \nonumber 
	&\cH_1:  	y_{\text{W},m}[l] = 
\begin{cases}
n_{\text{W},m}[l], ~ \forall m \notin \sv, \\
\sqrt{P_m} g_{\text{W},m} x_m[l] + n_{\text{W},m}[l], ~ \forall m \in \sv,
\end{cases}
\end{align}
where $x_m[l] \sim \cC\cN(0,1)$ indicates the data symbol transmitted by Alice through modality $m$ in channel use $l$.
Here, $P_m$ means the transmit power for modality $m$ and $g_{\text{W},m}$ is the channel between Alice and Willie for modality $m$. 
In addition, $n_{\text{W},m}[l] \sim \cC\cN(0,\Omega_m N_0)$ signifies the AWGN at Willie where $\Omega_m$ and $N_0$ represent the allocated bandwidth of modality $m$ and noise spectral density, respectively.

In this paper, we assume a highly sophisticated Willie who possesses full information about the channels from Alice to himself, the transmit powers, and the allocated bandwidths for all modalities.
This represents a worst-case scenario from a covertness perspective. 
Regarding Willie's knowledge of the channels from Alice to Bob, we consider two distinct cases: (i) Willie is aware of these channels, enabling him to infer the set of modalities $\sv$ selected by Alice; and (ii) Willie is uninformed of these channels, and consequently, remains ignorant of the selected set $\sv$.

The goal of Willie is to optimize his detection performance based on his knowledge of the wireless environment. 
By observing the received signals over $L$ channel uses, Willie makes a decision regarding the presence of a transmission. 
Let $\mathcal{D}_0$ and $\mathcal{D}_1$ denote the decisions in favor of $\mathcal{H}_0$ and $\mathcal{H}_1$, respectively. 
When the modality set $\sv$ is chosen by Alice, the DEP at Willie is expressed as
\begin{align} \label{eq:DEP_def_1}
    \tP_{\text{DEP},\sv} &= \tP_{\text{MD},\sv} + \tP_{\text{FA},\sv},   
\end{align}
where $\tP_{\text{MD},\sv} = \PP(  \cD_0 | \cH_1 )$ represents the missed detection probability and $\tP_{\text{FA},\sv} = \PP( \cD_1 | \cH_0 )$ accounts for the false alarm probability. 
The optimal detector at Willie is designed to minimize the DEP in~\eqref{eq:DEP_def_1}.

\begin{figure}
\begin{center} 
\includegraphics[width=3.3in]{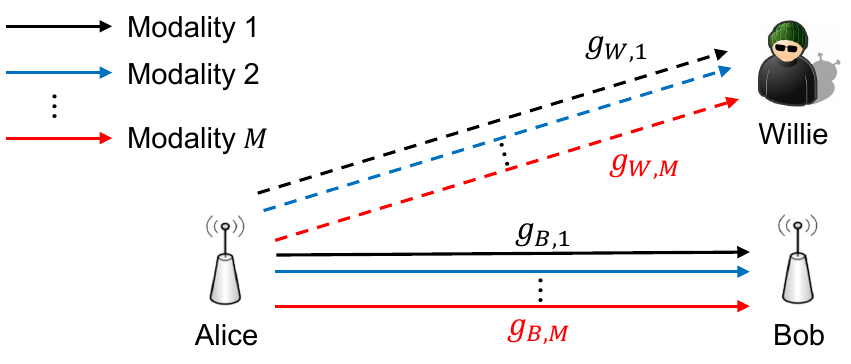}
\end{center} 
\vspace{-6.0mm}
\caption{Network model for simultaneous multi-modal covert communications.}
\label{figure:model}
\end{figure}

\subsection{Communication Performance at Bob}

For any given modality $m \in \sv$ and channel use $l$, the received signal at Bob becomes
\begin{align} \label{eq:Received_signal_Bob}
	y_{\text{B},m}[l] = \sqrt{P_m} g_{\text{B},m} x_m[l] + n_{\text{B},m}[l],
\end{align}
where $g_{\text{B},m}$ denotes the channel between Alice and Bob for modality $m$ and $n_{\text{B},m}[l] \sim \cC\cN(0,\Omega_m N_0)$ is the AWGN at Bob. 
Then, the SNR at Bob in modality $m \in \sv$ is written as 
\begin{align} \label{eq:SNR_Bob}
	\rho_{\text{B},m} = \frac{P_m \left| g_{\text{B},m} \right|^2 }{\Omega_m N_0}.
\end{align}

With the finite blocklength, the achieved rate (bps) at Bob for modality $m \in \sv$ is given by~\cite{Polyanskiy:10,Yan:19b}
\begin{align} \label{eq:rate_each_mode_finite}
	U_m \!=\!  \Omega_m \! \left[ \! \log_2 \left(1 \!+\! \rho_{\text{B},m} \right) \!-\! \frac{1}{\ln(2)} \! \sqrt{\frac{1\!\!-\!\left(1 \!+\! \rho_{\text{B},m}\right)^{-2}}{L}} Q^{-1}\!(\tau) \! \right]\!,
\end{align}
where $\tau$ is the acceptable outage probability and $Q^{-1}(\cdot)$ denotes the inverse of the Gaussian Q-function.
Then, the sum-rate at Bob is defined as
\begin{align} \label{eq:sum_rate}
	U_{\sv} = \sum_{m \in \sv} U_m.
\end{align}

\subsection{Optimization Problem}

We investigate the modality set selection problem to maximize the DEP at Willie while ensuring a minimum required rate at Bob. 
The optimization problem is formulated as
\begin{align} \label{eq:problem}
	&\underset{ \sv \subseteq \cM }{\max} ~~ \tP_{\text{DEP},\sv}  \\ \label{eq:rate_constraint}
	&\quad \text{s.t.} ~ U_{\sv} \geq U_{\text{target}},
\end{align} 
where $U_{\text{target}}$ represents the minimum rate requirement at Bob.



\section{Detection by a Fully-Aware Adversary}  \label{sec:analysis}

In this section, we analyze the DEP for the scenario where Willie is aware of the selected modality set~$\sv$. 
In this case, Willie measures the received signals only on the modalities in~$\sv$.

\subsection{Optimal Decision Rule}

The distribution of the received signal at Willie for $m \in \sv$ is given by 
\begin{align} \label{eq:Received_signal_Willie_distribution}
	&\cH_0:  	y_{\text{W},m}[l] \sim \cC\cN \left(0, \sigma_{0,m}^2 \right),  \\ \nonumber 
	&\cH_1:  	y_{\text{W},m}[l] \sim  \cC\cN \left(0,  \sigma_{1,m}^2 \right),
\end{align}
where
\begin{align} \label{eq:Received_signal_Willie_variance}
	&\sigma_{0,m}^2 =  \Omega_m N_0, \\
	& \sigma_{1,m}^2 =  P_m \left| g_{\text{W},m} \right|^2 + \Omega_m N_0. \nonumber 
\end{align}

Let $Y_{\text{W}}$ define the observation at Willie over all channel uses $l=1,\dots,L$ and modalities $m_1,\dots,m_{M_{\sv}} \in \sv$ where $M_{\sv} = |\sv|$.
We express the likelihood as
\begin{align} \label{eq:Likelihood}
	f(Y_{\text{W}} |\cH_k) = \prod_{ m \in \sv} \prod_{l=1}^L \frac{1}{ \pi \sigma_{k,m}^2 } \exp\bigg( - \frac{ \left| y_{\text{W},m}[l] \right|^2 }{  \sigma_{k,m}^2  } \bigg),
\end{align}
where $k \in \{0,1\}$. 
Then, the likelihood ratio test (LRT) is written as
\begin{align} \label{eq:LRT}
	\Lambda(Y_{\text{W}}) = \frac{f(Y_{\text{W}} |\cH_1)}{f(Y_{\text{W}} |\cH_0)} \underset{\cD_0}{\overset{\cD_1}{\gtrless}}  1.
\end{align}

\begin{theorem} \label{Theorem:optimal_detector_known_set}
When the selected subset of modalities $\sv$ is known to Willie, the optimal decision rule at Willie that minimizes the DEP $\tP_{\text{DEP},\sv}$ is
\begin{align} \label{eq:decision_criteria}
	T_{\sv} \underset{\cD_0}{\overset{\cD_1}{\gtrless}} \delta_{\sv},
\end{align} 
where the SNR at Willie $\rho_{\text{W},m}$, weight $w_m$, received signal strength $E_m$, test statistic $T_{\sv}$, and detection threshold $\delta_{\sv}$ are respectively expressed as 
\begin{align} \label{eq:SNR_W_m} 
    &\rho_{\text{W},m} = \frac{P_m \left| g_{\text{W},m} \right|^2}{\Omega_m N_0}, \\ \label{eq:def_weight}
	&w_m = \frac{\rho_{\text{W},m}}{ \Omega_m N_0 \left( 1+\rho_{\text{W},m} \right)   }, \\  \label{eq:def_Energy}
    &E_m =   \sum_{l=1}^L  \left| y_{\text{W},m}[l] \right|^2, \\ \label{eq:T_def}
    &T_{\sv} = \sum_{m \in \sv} w_m E_m, \\ \label{eq:def_threshold}
    &\delta_{\sv} = L \sum_{m \in \sv} \ln\left(  1 + \rho_{\text{W},m}     \right).
\end{align}
\end{theorem}
\begin{proof}
See Appendix~\ref{appendix:optimal_detector_known_set}.  
\end{proof}

Unlike a standard radiometer that simply sums the total received signal power, the optimal detector in~\eqref{eq:decision_criteria} strategically aggregates received signal power across active modalities by applying weights $\{ w_m \}$ to the energy measured in the modalities.
The weight for modality $m$ scales with the SNR $\rho_{\text{W},m}$, implying that Willie places greater emphasis on modalities where the signal is expected to be strong. 
Furthermore, Willie adjusts the detection threshold $\delta_{\sv}$ based on the SNRs to account for the varying received signal strengths.

\subsection{Exact DEP analysis}

Applying the detector in~\eqref{eq:decision_criteria}, the DEP in~\eqref{eq:DEP_def_1} is given by
\begin{align} \label{eq:DEP_def}
    \tP_{\text{DEP},\sv} &= \PP(  \cD_0 | \cH_1 ) + \PP(  \cD_1 | \cH_0 ) \nonumber \\
                     &= \PP\left(  T_{\sv} < \delta_{\sv} \Big| \cH_1 \right) + \PP\left(  T_{\sv} > \delta_{\sv} \Big| \cH_0 \right).
\end{align}
The exact mathematical expression of the DEP is identified in the following theorem.

\begin{theorem} \label{Theorem:exact_DEP}
The exact DEP with the optimal detector at Willie in~\eqref{eq:decision_criteria} is 
\begin{align} \label{eq:DEP_from_CF}
    \tP_{\text{DEP},\sv} = 1 \!-\! \frac{1}{\pi} \int_0^\infty \! \frac{1}{t} \Im \left[ e^{-it \delta_{\sv}} \left( \Phi_{T_{\sv} | \cH_1}(t) - \Phi_{T_{\sv} | \cH_0}(t) \right) \right] \d t,
\end{align}
where $\Im[x]$ denotes the imaginary part of $x$ and the characteristic function (CF) of the test statistic $T_{\sv}$ is
\begin{align} \label{eq:CF_of_T}
	\Phi_{T_{\sv}|\cH_k}(t) = \prod_{m \in \sv} \bigg( \frac{1}{1 - i t w_m \sigma_{k,m}^2} \bigg)^{L},
\end{align}
for $k \in \{0, 1\}$.
\end{theorem}
\begin{proof}
See Appendix~\ref{appendix:exact_DEP}.  
\end{proof}

In~\eqref{eq:DEP_from_CF}, the difference term $\Phi_{T_{\sv}|\cH_1}(t)-\Phi_{T_{\sv}|\cH_0}(t)$ represents the statistical separability of the received energy distributions under the two hypotheses.
The exponential term $e^{-it \delta_{\sv}}$ acts as a phase rotation determined by the detection threshold $\delta_{\sv}$, effectively aligning the decision boundary with these distributions.
In this sense, the DEP expression in~\eqref{eq:DEP_from_CF} can be interpreted as the total region of confusion, measuring the detector's fundamental performance limit.

\subsection{Closed-Form DEP Approximation} \label{subsec:approx_DEP_known}

Although the expression in~\eqref{eq:DEP_from_CF} yields the exact DEP, evaluating it incurs high computational complexity due to the need for numerical integration of the CFs. 
The main difficulty in analyzing the DEP arises from the complicated distribution of the test statistic $T_{\sv}$ in~\eqref{eq:T_def}, which is a weighted sum of independent random variables. 
To address this, we approximate $T_{\sv}$ as a single gamma variable via a two-moment matching method. 
In the following lemma, we introduce a low-complexity approximation of the DEP based on this approach.

\begin{lemma} \label{Lemma:approx_DEP}
The DEP $\tP_{\text{DEP},\sv}$ in~\eqref{eq:DEP_def_1} with the optimal detector in~\eqref{eq:decision_criteria} is approximated as 
\begin{align} \label{eq:DEP_moment_matching}
    \tP_{\text{DEP},\sv} &\approx 1 + \frac{\gamma\left( \alpha_{\sv,1}, \frac{\delta_{\sv}}{\theta_{\sv,1}} \right)}{\Gamma\left(\alpha_{\sv,1}\right)} - \frac{\gamma\left( \alpha_{\sv,0}, \frac{\delta_{\sv}}{\theta_{\sv,0}} \right)}{\Gamma\left(\alpha_{\sv,0}\right)},
\end{align}
where $\alpha_{\sv,k}$ and $\theta_{\sv,k}$ are respectively defined in~\eqref{eq:shape_approx} and~\eqref{eq:scale_approx} for $k \in \{ 0, 1 \}$. 
Here, $\Gamma(s)=\int_0^\infty t^{s-1} e^{-t} \d t$ and $\gamma(s,x) = \int_0^x t^{s-1} e^{-t} \d t$ are the gamma function and lower incomplete gamma function, respectively.
\end{lemma}
\begin{proof}
See Appendix~\ref{appendix:approx_DEP}.  
\end{proof}

Compared to the exact expression in~\eqref{eq:DEP_from_CF} which requires numerical integration of complex CFs, the approximation in~\eqref{eq:DEP_moment_matching} achieves significant computational efficiency by reducing the calculation to standard incomplete gamma functions, whose values can be rapidly obtained via table lookup or standard software functions. 

The approximation in Lemma~\ref{Lemma:approx_DEP} becomes identical to the exact DEP when $\{w_m\}$ and $\{\sigma_{k,m}\}$ are equal across modalities since the sum of $M_{\sv}$ independent and identically distributed (i.i.d.) gamma random variables $\cG(\alpha,\theta)$ is a gamma random variable $\cG( M_{\sv} \alpha,\theta)$.
The numerical simulations in Section~\ref{sec:results} validate that this approximation remains tight in the general case where $\{w_m\}$ and $\{\sigma_{k,m}\}$ vary across modalities.

\begin{corollary} 
As a special case, when only a single modality $m$ is selected by Alice, i.e., $M_{\sv} = 1$, the expression in~\eqref{eq:DEP_moment_matching} is equal to the exact DEP. 
In this case, the DEP is simplified to
\begin{align} \label{eq:DEP_single_mod}
    \tP_{\text{DEP},m} &= 1 - \frac{1}{\Gamma(L)} \left[  \gamma\left( L, \frac{\delta_m}{w_m \sigma_{0,m}^2} \right) - \gamma\left( L, \frac{\delta_m}{w_m  \sigma_{1,m}^2} \right)            \right] \nonumber \\
    &= 1 - \frac{1}{\Gamma(L)} \Bigg[  \gamma\left( L, L \left( 1 + \frac{ 1 }{ \rho_{\text{W},m} }  \right)\ln\left(  1 + \rho_{\text{W},m}    \right) \right) \nonumber\\
    &\qquad- \gamma\left( L, L  \frac{1 }{ \rho_{\text{W},m} } \ln\left(  1 + \rho_{\text{W},m}     \right) \right)     \Bigg],
\end{align}
where $\delta_m = L\ln\left(  1 + \rho_{\text{W},m}    \right)$.
The expression in~\eqref{eq:DEP_single_mod} is identical to the result in~\cite{Kong:22WCNC}. 
\end{corollary}

\begin{corollary} 
As a particular case, when the data symbols across all $M_{\sv}$ modalities in $\sv$ are identical, i.e., $x_m[l] = x[l], \forall m \in \sv$, the network model becomes equivalent to a scenario where Willie is equipped with $M_{\sv}$ antennas in terms of the DEP analysis. 
In this scenario, the optimal detector at Willie performs maximum ratio combining (MRC), and the corresponding SNR at Willie is $\rho_{\text{W},\sv}^{(\text{MRC})} = \sum_{m \in \sv} \rho_{\text{W},m}$. 
Then, in the same way as in~\eqref{eq:DEP_single_mod}, the resultant DEP is written as
\begin{align} \label{eq:DEP_MRC}
    \tP_{\text{DEP},\sv}^{(\text{MRC})} &= 1 \!-\! \frac{1}{\Gamma(L)} \Bigg[  \gamma\Bigg( L, L \Bigg( 1 + \frac{ 1 }{ \rho_{\text{W},\sv}^{(\text{MRC})} }  \Bigg)\ln\left(  1 + \rho_{\text{W},\sv}^{(\text{MRC})}    \right) \!\! \Bigg) \nonumber \\
    &\qquad \quad - \gamma\Bigg( L, L  \frac{1 }{ \rho_{\text{W},\sv}^{(\text{MRC})} } \ln\left(  1 + \rho_{\text{W},\sv}^{(\text{MRC})}    \right) \Bigg)     \Bigg].
\end{align}
\end{corollary}

\subsection{Conventional DEP Metrics}

In this subsection, we introduce two conventional approaches for analyzing the DEP performance in simultaneous multi-modal transmission. 

In~\cite{Lin:23}, based on Pinsker's inequality, a lower bound of the DEP is given by
\begin{align} \label{eq:DEP_KL_bound}
    \tP_{\text{DEP},\sv} &\geq 1 - \sqrt{\frac{1}{2}\cD(\PP_{0} | \PP_{1})} \nonumber \\
    &= 1 - \sqrt{\frac{L}{2}  \sum_{m \in \sv}  \left[  \ln\left( 1 + \rho_{\text{W},m} \right)  - \frac{1}{1 + 1/\rho_{\text{W},m} }      \right]  } \nonumber \\
    &\triangleq \tP_{\text{DEP},\sv}^{\text{(LB)}},
\end{align}
where $\PP_{0}$ and $\PP_{1}$ are the distributions of the received signals at Willie under $\cH_0$ and $\cH_1$, respectively.


Also, in~\cite{Aggarwal:24}, for the case where $M_{\sv} = 2$, an approximation of the DEP using an iterative detection threshold optimization was studied.  
We generalize this work to scenarios with any number of selected modalities $M_{\sv}$.
Assume Willie employs $M_{\sv}$ detection thresholds, $\delta_1, \dots \delta_{m_{M_{\sv}}}$.
For each modality, the false alarm probability and missed detection probability are given by $\tP_{\text{FA},m} = 1 - \frac{1}{\Gamma\left(L\right)}\gamma\left( L, \frac{\delta_m}{  \Omega_m N_0  } \right)$ and $\tP_{\text{MD},m} =  \frac{1}{\Gamma\left(L\right)} \gamma\Big( L, \frac{\delta_m}{   P_m \left| g_{\text{W},m} \right|^2 + \Omega_m N_0  } \Big)$, respectively. 
A detection error occurs if there is at least one false alarm or if Willie fails to detect transmissions at all modalities.
Therefore, the DEP is defined by 
\begin{align} \label{eq:DEP_Iterative_def}
    \tP_{\text{DEP},\sv}^{\text{(iter)}}(\{ \delta_m\}) &= 1 - \prod_{m\in\sv} \left( 1 - \tP_{\text{FA},m} \right) + \prod_{m\in\sv} \tP_{\text{MD},m}.
\end{align}
By computing the partial derivative of $\tP_{\text{DEP},\sv}^{\text{(iter)}}$ with respect to $\delta_m$ and equating it to zero, we have
\begin{align} \label{eq:delta_Iterative}
    \delta_m &= \Omega_m N_0 L \bigg( 1 + \frac{1}{\rho_{\text{W},m} } \bigg)  \Bigg[ \ln\left( 1 + \rho_{\text{W},m} \right) \nonumber \\
    &\qquad+ \frac{1}{L} \ln\Bigg(  \frac{ \prod_{ \bar{m}\in\sv, \bar{m} \neq m} ( 1 - \tP_{\text{FA},\bar{m}} ) }{  \prod_{\bar{m}\in\sv, \bar{m} \neq m} \tP_{\text{MD},\bar{m}}   }        \Bigg) \Bigg].
\end{align}
Finally, we update $\delta_1, \dots \delta_{m_{M_{\sv}}}$ iteratively until $\tP_{\text{DEP},\sv}^{(\text{iter})}$  converges.
The corresponding approximation of the DEP can be obtained by substituting the converged thresholds into~\eqref{eq:DEP_Iterative_def}.

Numerical simulations in Section~\ref{sec:results} show that the lower bound $\tP_{\text{DEP},\sv}^{\text{(LB)}}$ is overly conservative, while $\tP_{\text{DEP},\sv}^{\text{(iter)}}$ overestimates the actual DEP despite its prohibitively high computational complexity.


\section{Detection under Modality Uncertainty}  \label{sec:unknown}

In Section~\ref{sec:analysis}, we focused on the worst-case scenario where Willie has full knowledge of the transmission parameters to optimize his detection performance. 
In this section, we consider a more practical scenario where Willie possesses no knowledge of the channels between Alice and Bob. 
Since Alice's objective is to maximize covertness subject to a minimum rate requirement at Bob, her selection of modalities is jointly determined by the channel conditions of both the Alice-Willie and Alice-Bob links.
Consequently, as Willie lacks access to the channels between Alice and Bob, he is fundamentally unable to infer the specific subset of modalities selected by Alice.

\subsection{Optimal Decision Rule under Uncertainty}

In contrast to the scenario in Section~\ref{sec:analysis} where the selected subset $\sv$ is known to Willie, here he must process the observations across all available modalities $m \in \mathcal{M}$, resulting in increased detection uncertainty.

\begin{theorem} \label{Theorem:optimal_detector_unknown}
When Willie is uninformed of the selected set of modalities $\sv$ and adopts a uniform prior distribution over all valid subsets, the optimal decision rule for Willie is given by
\begin{align} \label{eq:LRT_unknown_final}
    T \triangleq \eta \sum_{ \vv \subseteq \mathcal{M}, \vv \neq \emptyset} \exp\left(  T_{\vv} - \delta_{\vv} \right) \underset{\mathcal{D}_0}{\overset{\mathcal{D}_1}{\gtrless}}  1,
\end{align}
where $\eta \triangleq 1 / (2^M-1)$, $T_{\vv} =  \sum_{m \in \vv} w_m E_m$ is the weighted sum of energies across the modalities in a set $\vv \subseteq \cM$, and $\delta_{\vv} \triangleq L \sum_{m \in \vv} \ln \left(1 + \rho_{\text{W},m} \right)$.
\end{theorem}
\begin{proof}
See Appendix~\ref{appendix:optimal_detector_unknown}.  
\end{proof}

When the selected set of modalities $\sv$ is known to Willie, the optimal detector in~\eqref{eq:decision_criteria} accumulates energy linearly and exclusively over the active modalities. In contrast, the optimal detector in~\eqref{eq:LRT_unknown_final} computes the arithmetic mean of the likelihood ratios for every potential subset configuration $\vv \subseteq \cM$. This operation forces the detector to integrate noise from the inactive modalities, thereby diluting the accumulated signal energy and degrading detection performance relative to the scenario in~\eqref{eq:decision_criteria}.

\subsection{DEP Analysis in Low-SNR Regime}

If Willie is aware of the selected subset $\sv$ as in Section~\ref{sec:analysis}, the test statistic $T_{\sv}$ in~\eqref{eq:T_def} is a linear sum of the energies across modalities, which is analytically tractable. 
On the other hand, when $\sv$ is unknown to Willie, the test statistic $T$ in~\eqref{eq:LRT_unknown_final} becomes a sum of exponentials of the weighted energy sums $T_{\vv} =  \sum_{m \in \vv} w_mE_m$.  
Therefore, deriving the CF of this test statistic and the corresponding DEP in closed form is mathematically intractable.

In covert communications, since Alice aims to conceal the existence of her transmissions, the received signal power at Willie must be maintained at a very low level. 
Motivated by this operational requirement, the following theorem derives an expression for the DEP in the low-SNR regime at Willie.

\begin{theorem} \label{Theorem:lowSNR_DEP}
When Willie is unaware of the set of active modalities $\sv$, the DEP in the low-SNR regime ($|\rho_{\text{W},m}| \ll 1$) is written as
\begin{align} \label{eq:DEP_unknown}
    \tP_{\text{DEP},\sv} =  1 - \frac{1}{\pi} \int_0^\infty \! \frac{1}{t} \Im \left[ e^{-it \tilde{\delta}} \left(  \Phi_{\tilde{T} | \cH_1 }(t) - \Phi_{\tilde{T} | \cH_0 }(t) \right) \right] \d t,
\end{align}
where $\tilde{T} = \sum_{m \in \cM} w_mE_m$ denotes the test statistic and $\tilde{\delta} =  L \sum_{m \in \cM} \ln \left(1 + \rho_{\text{W},m} \right)$ is the detection threshold in the low-SNR regime. 
Here, the CFs of $\tilde{T}$ under two hypotheses are represented as 
\begin{align} \label{eq:CF_of_T_H0}
	\Phi_{\tilde{T} | \cH_0 }(t) &= \prod_{m \in \cM} \phi_{0,m}(t),   \\ \label{eq:CF_of_T_H1}
    \Phi_{\tilde{T} | \cH_1 }(t) &= \left( \prod_{m \in \sv} \phi_{1,m}(t) \right) \left( \prod_{m \notin \sv} \phi_{0,m}(t) \right),
\end{align}
where $\phi_{k,m}(t)$ for $k \in \{0,1\}$ is defined in~\eqref{eq:phi_def}. 
\end{theorem}
\begin{proof}
See Appendix~\ref{appendix:lowSNR_DEP}.  
\end{proof}

The CFs in~\eqref{eq:CF_of_T_H0} and~\eqref{eq:CF_of_T_H1} highlight the impact of the ratio of the total available modalities $M$ to the number of active modalities $M_{\sv}$ on the DEP.
When $\sv$ is unknown, Willie is compelled to integrate energy over all $M$ available modalities to avoid missing the signal, thereby unavoidably collecting noise from the empty frequency bands.
This accumulated noise increases the variance of the test statistic $\tilde{T}$, effectively obscuring the signal energy present in the active subset $\sv$.

\subsection{Low-Complexity Approximation}

Since the representation of the DEP in~\eqref{eq:DEP_unknown} incurs high computational complexity, in the following lemma, we introduce a low-complexity approximation of the DEP based on the two-moment matching method.

\begin{lemma} \label{Lemma:approx_DEP_unknown}
An approximation of the DEP $\tP_{\text{DEP},\sv}$ in~\eqref{eq:DEP_unknown} in the low-SNR regime is identified as 
\begin{align} \label{eq:DEP_moment_matching_unknown}
    \tP_{\text{DEP},\sv} &\approx  1 + \frac{\gamma\left( \tilde{\alpha}_{1}, \frac{\tilde{\delta}}{ \tilde{\theta}_{1}} \right)}{\Gamma\left(\tilde{\alpha}_{1}\right)} - \frac{\gamma\left( \tilde{\alpha}_{0}, \frac{ \tilde{\delta}}{\tilde{\theta}_{0}} \right)}{\Gamma\left(\tilde{\alpha}_{0}\right)},
\end{align}
where $\tilde{\alpha}_{k}$ and $\tilde{\theta}_{k}$ are defined in Appendix~\ref{appendix:approx_DEP_unknown}.
\end{lemma}
\begin{proof}
See Appendix~\ref{appendix:approx_DEP_unknown}.  
\end{proof}

The approximation in~\eqref{eq:DEP_moment_matching_unknown} explicitly quantifies the noise aggregation penalty via the effective shape ($\tilde{\alpha}_{k}$) and scale ($\tilde{\theta}_{k}$) parameters, allowing for rapid DEP evaluation without requiring numerical integration.

\subsection{Lower Bound of the DEP}

When the set of selected modalities $\sv$ is unknown to Willie, he monitors all $M$ modalities. 
Consequently, the total divergence $\cD(\PP_{0} | \PP_{1})$ in~\eqref{eq:DEP_KL_bound} becomes the sum of divergences across all $M$ modalities: $\cD(\PP_{0} | \PP_{1}) = \sum_{m\in\sv} \cD\left( \PP_{0,m} | \PP_{1,m} \right) + \sum_{m\notin\sv} \cD\left( \PP_{0,m} | \PP_{1,m} \right)$.
However, for $m\notin \sv$, $\PP_{0,m}$ is identical to $\PP_{1,m}$, implying that the second summation term is zero. 
Therefore, the lower bound of the DEP $\tP_{\text{DEP},\sv}^{\text{(LB)}}$ remains identical to the expression derived in~\eqref{eq:DEP_KL_bound}.


\section{Modality Set Selection for DEP Maximization} \label{sec:optimization}

In this section, we focus on the modality set selection problem in~\eqref{eq:problem}, which aims to maximize the DEP under a rate requirement at Bob. 
The optimization is a combinatorial problem; since the search space size of $2^M - 1$ grows exponentially with the number of modalities, finding the global optimum via exhaustive search is computationally prohibitive for large $M$. 
We propose a low-complexity subset selection strategy in Section~\ref{subsec:proposed} for the case where Alice knows the instantaneous CSI to Willie, and subsequently consider the scenarios where Alice possesses only statistical CSI or no CSI regarding the channels to Willie in Sections~\ref{subsec:statisticalCSI} and~\ref{subsec:noCSI}, respectively.

\subsection{Proposed algorithm} \label{subsec:proposed}

Based on the detector analysis, the test statistic is a linear combination of weighted energies for the case where the modalities are known ($T_{\sv} = \sum_{m \in \sv} w_m E_m$) and the case where modalities are unknown in the low-SNR regime ($\tilde{T} = \sum_{m \in \cM} w_mE_m$). 
Consequently, the statistical distance between the null and alternate hypotheses is primarily determined by the aggregate shift in the mean of the test statistic. 
In this sense, we employ the mean shift as a tractable metric to be minimized. 
However, this approach entails a loss of optimality since it does not directly optimize the exact DEP expression in~\eqref{eq:DEP_from_CF}, which is a non-linear function of the decision statistics.

We define $c_m$ as the cost of modality $m$, which quantifies its contribution to the shift in the mean of the test statistic:
\begin{align} \label{eq:cost_metric}
    c_m &= \mathbb{E}\left[w_mE_m \big| \mathcal{H}_1 \right] - \mathbb{E}\left[w_mE_m \big| \mathcal{H}_0\right]  \\
    &=  L w_m \left( \sigma_{1,m}^2 - \sigma_{0,m}^2 \right)  = L w_m P_m |g_{\text{W},m}|^2 =  \frac{ L \rho_{\text{W},m}^2}{  1 + \rho_{\text{W},m} }. \nonumber
\end{align}
Then, we reformulate the modality selection problem as the following mean shift minimization problem:
\begin{align} \label{eq:problem_min_shift}
	&\underset{ \sv \subseteq \cM }{\min} ~~ \sum_{m \in \sv} c_m  \\ \label{eq:rate_constraint_repeat}
	&\quad \text{s.t.} ~ \sum_{m \in \sv} U_m \geq U_{\text{target}}.
\end{align} 
This formulation corresponds to the min-Knapsack problem (MKP), where the objective is to satisfy the rate constraint $U_{\text{target}}$ while minimizing the aggregate mean shift at Willie.

The optimization in~\eqref{eq:problem_min_shift} is known to be an NP-hard problem. 
Therefore, we propose a greedy approach based on the rate-to-cost efficiency metric, defined as
\begin{align} \label{eq:efficiency_metric}
    \Psi_m \triangleq \frac{U_m}{c_m} =  \frac{ U_m \left( 1 + \rho_{\text{W},m} \right) }{   L \left( \rho_{\text{W},m} \right)^2}.
\end{align}
Here, $\Psi_m$ quantifies the rate gain per unit of covertness cost. 
Since our goal is to accumulate rate to meet the requirement $U_{\text{target}}$ with the minimum loss of covertness, we prioritize modalities that offer the largest efficiency $\Psi_m$.
To solve the problem in~\eqref{eq:problem_min_shift}, we sort the modalities in descending order of $\Psi_m$ and sequentially add them to the set $\sv$ until the rate constraint $U_{\text{target}}$ is satisfied. 
Finally, to mitigate the potential sub-optimality of the greedy approach, we employ a refinement step that iteratively swaps a selected modality with an unselected one if the swap reduces the total covertness cost while maintaining the target rate requirement.
The proposed algorithm is summarized in Algorithm~\ref{Algorithm:proposed}.

\begin{algorithm} 
\caption{Proposed Modality Set Selection Algorithm }
\label{Algorithm:proposed}
\begin{algorithmic}[1] 
\State \textbf{Input:} $\mathcal{M}$, $U_{\text{target}}$
\State \textbf{Initialize:} $\sv \gets \emptyset$, $U_{\sv} \gets 0$

\State \textbf{Step 1: Metric Calculation}
\State Calculate the metric $\Psi_m$ in~\eqref{eq:efficiency_metric} for all $m \in \mathcal{M}$

\State \textbf{Step 2: Sorting}
\State Sort modalities in descending order of $\Psi_m$ such that indices $\hat{m}_1, \hat{m}_2, \dots, \hat{m}_M$ satisfy $\Psi_{\hat{m}_1} \ge \Psi_{\hat{m}_2} \ge \dots \ge \Psi_{\hat{m}_M}$

\State \textbf{Step 3: Greedy Selection}
\For{$i = 1$ to $M$}
    \State $U_{\sv} \gets U_{\sv} + U_{\hat{m}_i}$
    \State $\sv \gets \sv \cup \{\hat{m}_i\}$
    \If{$U_{\sv} \ge U_{\text{target}}$}
        \State \textbf{Break} 
    \EndIf
\EndFor

\State \textbf{Step 4: Refinement}
\State Define unselected set $\mathcal{R} \gets \mathcal{M} \setminus \sv$
\For{$m_{\text{out}} \in \sv$}
    \For{$m_{\text{in}} \in \mathcal{R}$}
        \State $U_{\text{new}} \gets U_{\sv} - U_{m_{\text{out}}} + U_{m_{\text{in}}}$
        \State $\Delta c \gets c_{m_{\text{in}}} - c_{m_{\text{out}}}$        
        \If{$U_{\text{new}} \ge U_{\text{target}}$ \textbf{and} $\Delta c < 0$}
            \State $\sv \gets (\sv \setminus \{m_{\text{out}}\}) \cup \{m_{\text{in}}\}$
            \State $U_{\sv} \gets U_{\text{new}}$
            \State $\mathcal{R} \gets (\mathcal{R} \setminus \{m_{\text{in}}\}) \cup \{m_{\text{out}}\}$
        \EndIf
    \EndFor
\EndFor

\If{$U_{\sv} \geq U_{\text{target}}$}
    \State \textbf{Output:} Optimized modality subset $\sv$
\Else
    \State \textbf{Output:} Infeasible
\EndIf

\end{algorithmic}
\end{algorithm}

The computational complexity of the proposed algorithm consists of the sorting operation, which requires a complexity of $\mathcal{O}(M \log M)$, and the refinement step, which incurs a complexity of $\mathcal{O}(|\sv| \cdot |\mathcal{R}|) \le \mathcal{O}(M^2)$, where $\mathcal{R}=\mathcal{M}\backslash \sv$, due to the pairwise inspection of selected and unselected modalities.
Consequently, the overall complexity is bounded by $\mathcal{O}(M^2)$.
Even with this quadratic bound, the complexity remains significantly lower than the exponential complexity $\mathcal{O}(2^M)$ required for the optimal exhaustive search.
Furthermore, in contrast to the exhaustive search, which requires evaluating the complex exact DEP expressions for every candidate subset, the proposed technique relies solely on the computation of the closed-form metric in~\eqref{eq:efficiency_metric}, significantly reducing the computational burden.

\subsection{Modality Set Selection with Statistical CSI} \label{subsec:statisticalCSI}

Let us denote $g_{\text{W},m} = \tilde{g}_{\text{W},m} \sqrt{ PL_m(d_{\text{W}})}$ where $\tilde{g}_{\text{W},m}$ is the small-scale fading of modality $m$, $d_{\text{W}}$ stands for the distance from Alice to Willie, and $PL_m(d)$ represents the large-scale path loss at distance $d$ for modality $m$. 
We consider the scenario where Alice does not know the instantaneous channel to Willie but possesses statistical CSI, specifically the large-scale path loss $\sqrt{ PL_m(d_{\text{W}})}$ and the statistics of the small-scale fading channels $\{ \tilde{g}_{\text{W},m} \}$.

Since the cost in~\eqref{eq:cost_metric} is $ c_m = L \frac{  \rho_{\text{W},m}^2}{  1 + \rho_{\text{W},m} } =  L \big(  \rho_{\text{W},m} - 1 + \frac{1}{1 + \rho_{\text{W},m}} \big)$, the expected cost $\bar{c}_m$ is given by
\begin{align} \label{eq:cost_CDI}
    \bar{c}_m =  L\bigg( \EE\bigg[\rho_{\text{W},m} \bigg] + \EE\bigg[ \frac{1}{1 + \rho_{\text{W},m}} \bigg] - 1  \bigg).
\end{align}
When $|\tilde{g}_{\text{W},m}|$ follows a Nakagami fading distribution with parameter $\kappa_m$, the SNR at Willie $\rho_{\text{W},m}$ follows a gamma distribution with shape $\kappa_m$ and scale $\bar{\rho}_{\text{W},m} / \kappa_m$ where $\bar{\rho}_{\text{W},m} \triangleq \EE \left[ \rho_{\text{W},m} \right] = \frac{P_m d_{\text{W}}^{-\epsilon}}{\Omega_m N_0}$.
Furthermore, we have
\begin{align} \label{eq:expected_term_2}
     &\EE\bigg[ \frac{1}{1 + \rho_{\text{W},m}} \bigg] = \int_0^\infty \frac{1}{1+x} f_{\rho_{\text{W},m}}(x) dx \nonumber\\
     & = \bigg(\frac{\kappa_m}{\bar{\rho}_{\text{W},m}}\bigg)^{\kappa_m} \int_0^\infty \frac{1}{1+x}   \frac{1}{\Gamma(\kappa_m)}  x^{\kappa_m-1} \exp\bigg(-\frac{\kappa_m x}{\bar{\rho}_{\text{W},m}}\bigg)     dx \nonumber \\
     & = \bigg(\frac{\kappa_m}{\bar{\rho}_{\text{W},m}}\bigg)^{\kappa_m} \cU\bigg(\kappa_m, \kappa_m, \frac{\kappa_m}{\bar{\rho}_{\text{W},m}}\bigg),
\end{align}
where $\cU(x,y,z) = \frac{1}{\Gamma(x)}\int_0^\infty e^{-zt} t^{x-1} (1+t)^{y-x-1} dt$ is the confluent hypergeometric function of the second kind. 
Consequently, the expected cost $\bar{c}_m$ in~\eqref{eq:cost_CDI} can be expressed as
\begin{align} \label{eq:cost_CDI_final}
    \bar{c}_m \!=\!  L\left( \bar{\rho}_{\text{W},m} \!+\! \bigg(\frac{\kappa_m}{\bar{\rho}_{\text{W},m}}\bigg)^{\kappa_m} \! \cU\bigg(\kappa_m, \kappa_m, \frac{\kappa_m}{\bar{\rho}_{\text{W},m}}\bigg) \!-\! 1  \right).
\end{align}

In the special case of Rayleigh fading ($\kappa_m = 1$), the expected cost simplifies to
 \begin{align} \label{eq:cost_CDI_final_Rayleigh}
    \bar{c}_m \!=\! L\bigg( \bar{\rho}_{\text{W},m} \!+\! \frac{1}{\bar{\rho}_{\text{W},m}} \exp\bigg( \frac{1}{\bar{\rho}_{\text{W},m}} \bigg) \text{E}_1\bigg( \frac{1}{\bar{\rho}_{\text{W},m}}\bigg) \!-\! 1  \bigg),
\end{align}
where we utilize the identity $\cU(1,1,z) = \exp(z)\text{E}_1(z)$ with $\text{E}_1(z) = \int_1^\infty \frac{\exp(-z t)}{t} \d t$ denoting the exponential integral.

Finally, to optimize modality selection under these statistical constraints, we define the efficiency metric $\Psi_m$ for Algorithm~\ref{Algorithm:proposed} as
\begin{equation} \label{eq:stat_metric}
    \Psi_m = \frac{U_m}{\bar{c}_m}.
\end{equation}

\subsection{Modality Set Selection without CSI} \label{subsec:noCSI}

In the scenario where Alice lacks both instantaneous and statistical CSI for the channels to Willie, she cannot explicitly evaluate the covertness cost. 
Since Willie's detection performance is fundamentally determined by the aggregate received energy, and Alice has no knowledge of the channel gains $|g_{\text{W},m}|^2$, a robust approach is to maximize energy efficiency while satisfying the rate requirement at Bob.
Accordingly, we define the energy efficiency metric for Algorithm~\ref{Algorithm:proposed} as
\begin{align} \label{eq:metric_no_csi}
    \Psi_m = \frac{U_m}{P_m}.
\end{align}
The proposed metric ensures that the rate constraint is met with the minimum possible total transmit power, offering a robust baseline for covertness in the absence of CSI.

\section{Simulation Results and Discussion} \label{sec:results}

In this section, we provide extensive numerical simulation results to validate our analysis and evaluate the efficiency of the proposed modality subset selection technique.
We consider 10 communication modalities $m_1,\dots, m_{10}$ where the center frequency of modality $m_i$ is given by $f_{c,i} = 300 \times i$~MHz. 
The channels from Alice to Willie and Bob in modality $m$ are modeled as $g_{\text{W},m} = \tilde{g}_{\text{W},m} \sqrt{ PL_m(d_{\text{W}})}$ and $g_{\text{B},m} = \tilde{g}_{\text{B},m} \sqrt{ PL_m(d_{\text{B}})}$, respectively, where $d_{\text{W}}$ and $d_{\text{B}}$ represent the distances from Alice to Willie and Bob, respectively.
Here, we employ the large-scale path loss model (in dB) $PL_m(d) = -10 n \log_{10}(d) - 20 \log_{10}(4\pi / \lambda_m)$~dB~\cite{Perez:02} where $n$ is a parameter that captures the effects of propagation mechanisms and $\lambda_m$ is the wavelength in meters. 
Furthermore, $\tilde{g}_{\text{W},m}$ and $\tilde{g}_{\text{B},m}$ denote the small-scale fading channels from Alice to Willie and Bob in modality $m$, respectively.

Unless otherwise stated, we adopt the following default parameters: $M=10$, $\cM = \{ m_1, \dots, m_{10}\}$, $\Omega_m = 10$~MHz, $\forall m$, $P_m = 10$~mW, $\forall m$, $N_0 = -120$~dBm, $L = 100$, $\tau = 0.15$, and $d_{\text{B}} = 30$~m.
We adopt the Rayleigh fading model; specifically, $\tilde{g}_{\text{W},m}$ and $\tilde{g}_{\text{B},m}$ are modeled as i.i.d. circularly symmetric complex Gaussian random variables with zero mean and unit variance, i.e., $\cC\cN(0,1)$. 
All analytical and simulated results are averaged over $10^4$ independent random channel realizations.

\begin{figure}[t]
	\centering
	\subfigure[The DEP at Willie]{  \hspace{-5mm}
		\includegraphics[width=.45\textwidth]{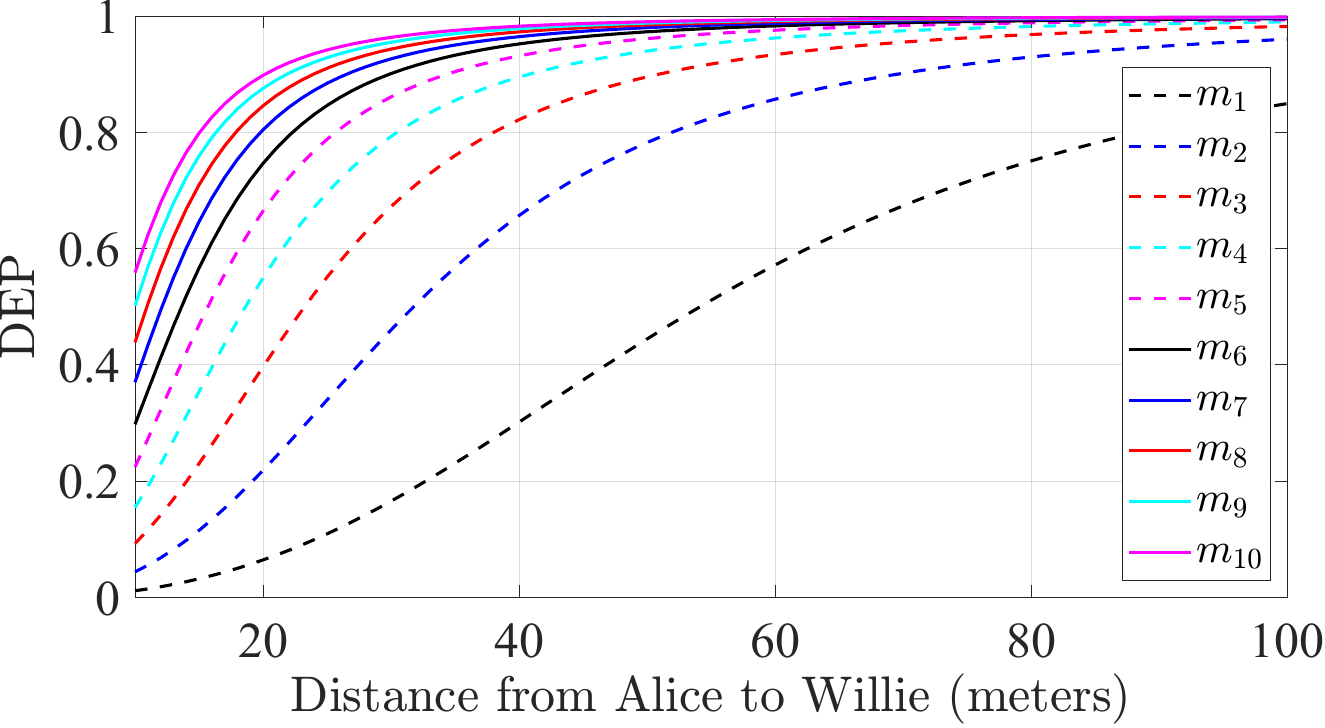} 
	}
	\vfill
	\subfigure[The rate at Bob]{ \hspace{-5mm}
		\includegraphics[width=.45\textwidth]{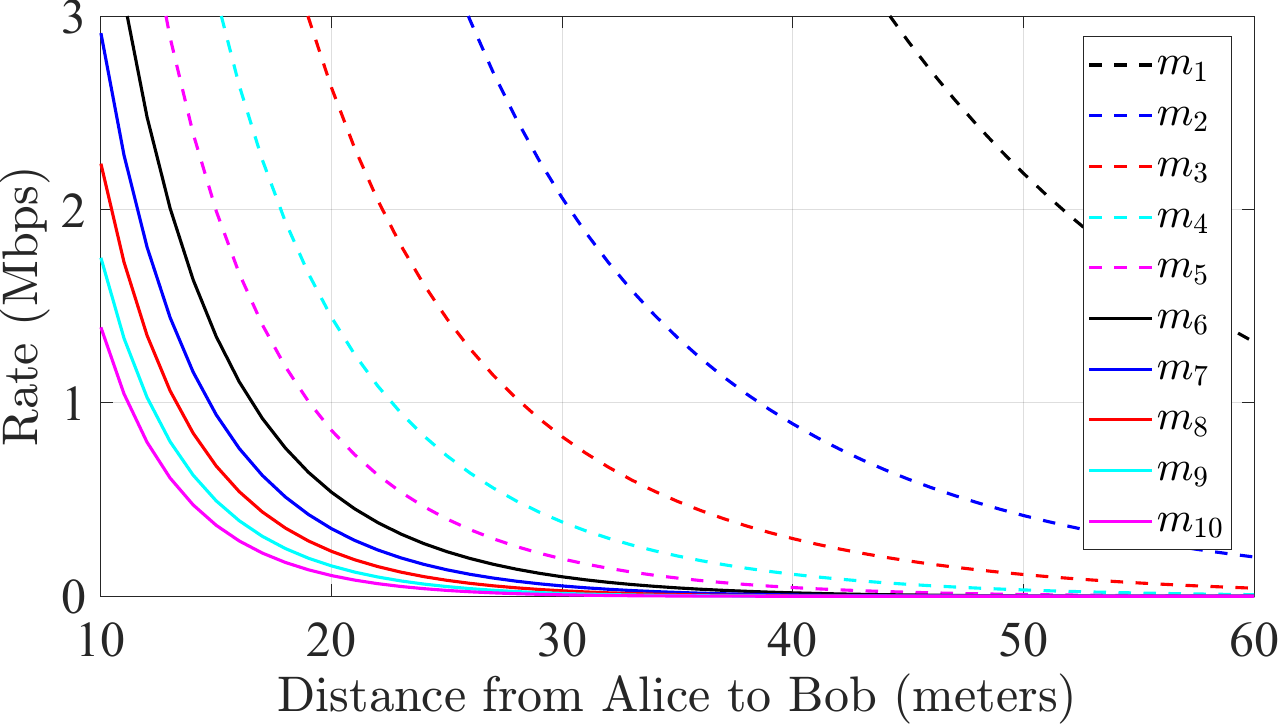}
	}
    \vspace{-2mm}
	\caption{Performance trade-off: DEP and rate for individual modalities.}
	\label{fig:each_modality}
\end{figure}

Fig.~\ref{fig:each_modality} demonstrates the DEP at Willie and the rate at Bob for each of the 10 available modalities. Since the center frequency of modality $m_i$ is $f_{c,i} = 300 i$~MHz, the channels from Alice to both Willie and Bob experience severe path loss as the index $i$ increases. 
Consequently, modalities with higher center frequencies yield higher DEP (improved covertness) but lower data rates. 
This highlights the fundamental trade-off in the selection process: modalities that are best for covertness contribute the least toward satisfying the rate requirement in an average sense.

\subsection{DEP when the active modalities are known to Willie}

\begin{figure}[t]
\begin{center} 
\includegraphics[width=3.3in]{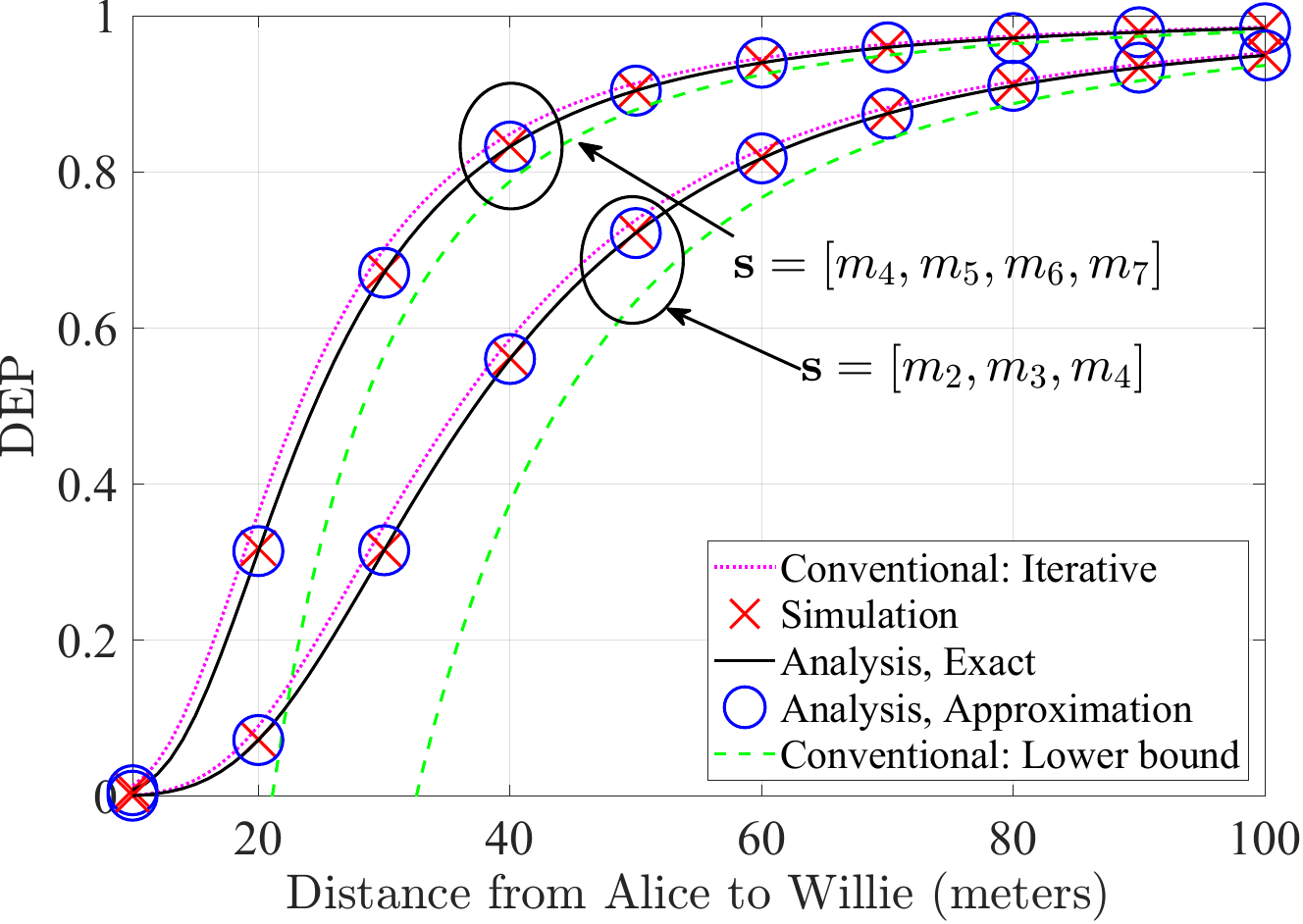}
\end{center} 
\vspace{-5.0mm}
\caption{Validation of the derived exact DEP expression and approximation.}
\label{figure:Known_1}
\end{figure}

In Fig.~\ref{figure:Known_1}, we validate the DEP analysis presented in Section~\ref{sec:analysis}. 
First, we observe that the exact analysis in~\eqref{eq:DEP_from_CF} accurately predicts the actual DEP, as evidenced by its perfect agreement with the Monte Carlo simulation results. 
Furthermore, the proposed approximation in~\eqref{eq:DEP_moment_matching} proves to be very tight, despite its significantly reduced computational complexity.
In contrast, the conventional approach using the iterative detection threshold optimization $\tP_{\text{DEP},\sv}^{\text{(iter)}}$ in~\eqref{eq:DEP_Iterative_def} overestimates the exact DEP. 
Therefore, relying on this metric for optimization may lead to insufficient actual covertness.
Conversely, the conventional lower bound $\tP_{\text{DEP},\sv}^{\text{(LB)}}$ in~\eqref{eq:DEP_KL_bound} exhibits a significant gap by underestimating the DEP. 
Hence, employing this lower bound as a metric would result in an overly conservative optimization, unnecessarily limiting transmission opportunities.

\begin{figure}[t]
\begin{center} 
\includegraphics[width=3.3in]{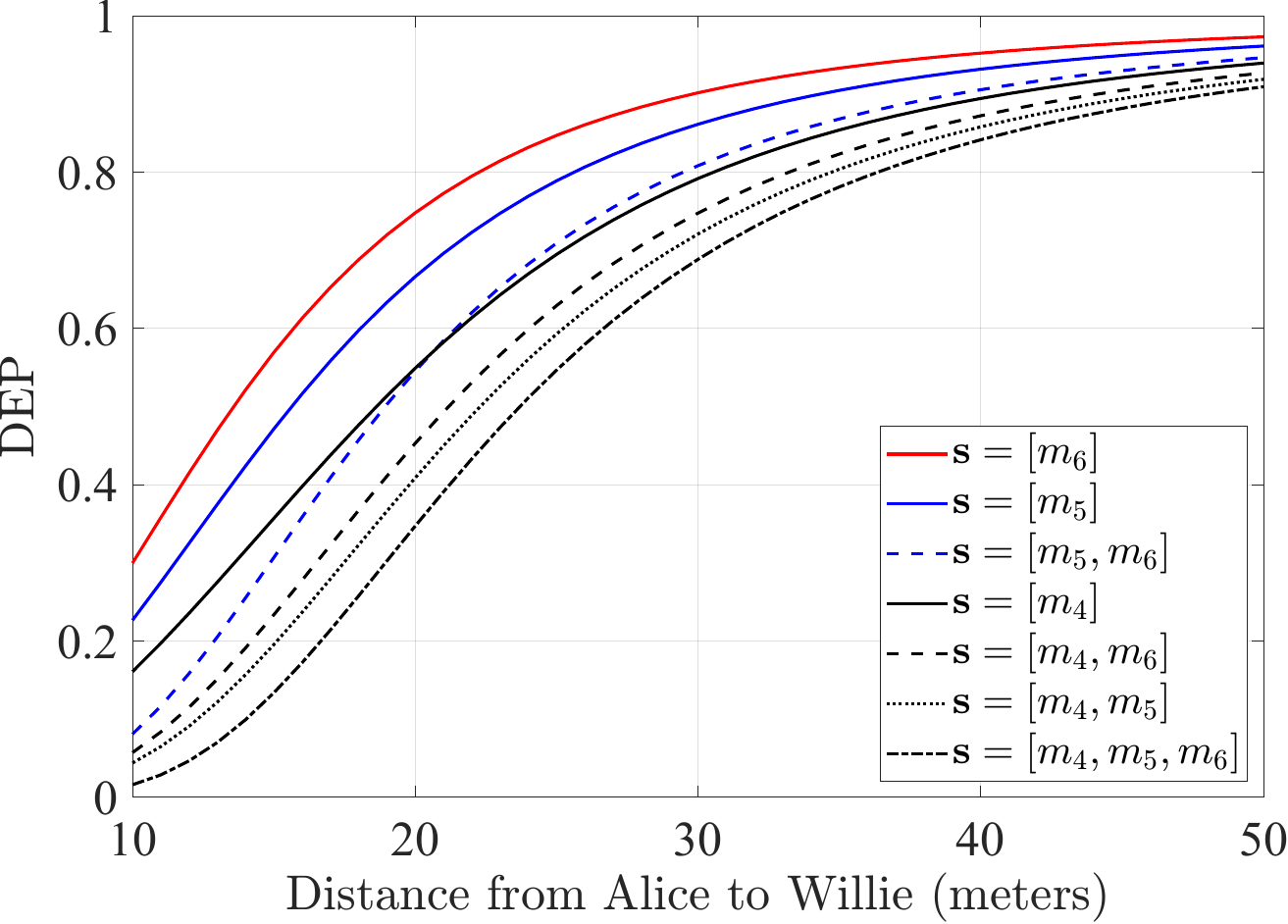}
\end{center} 
\vspace{-5.0mm}
\caption{DEP performance for various active modality subsets.}
\label{figure:Known_2}
\end{figure}

Fig.~\ref{figure:Known_2} presents the exact DEP performance for various modality subsets. 
Consistent with the results in Fig.~\ref{fig:each_modality}, single modalities with higher center frequencies exhibit higher DEP.
We observe that the achieved DEP decreases as we include additional modalities in the active set, since utilizing more channels increases the aggregate energy received at Willie. 
In this sense, the DEP is monotonically decreasing with respect to set inclusion (i.e., adding a modality to a set never increases the DEP).
However, the DEP is not a strictly decreasing function of the number of selected modalities. 
For instance, the set $\sv = \{m_5, m_6\}$ achieves a higher DEP than the single modality $\sv = \{m_4\}$ when the distance to Willie exceeds 20~m. 
This demonstrates that a strategic combination of multiple high-frequency modalities can be more covert than a single low-frequency modality.

\subsection{DEP when the active modalities are unknown to Willie}

\begin{figure}[t]
\begin{center} 
\includegraphics[width=3.3in]{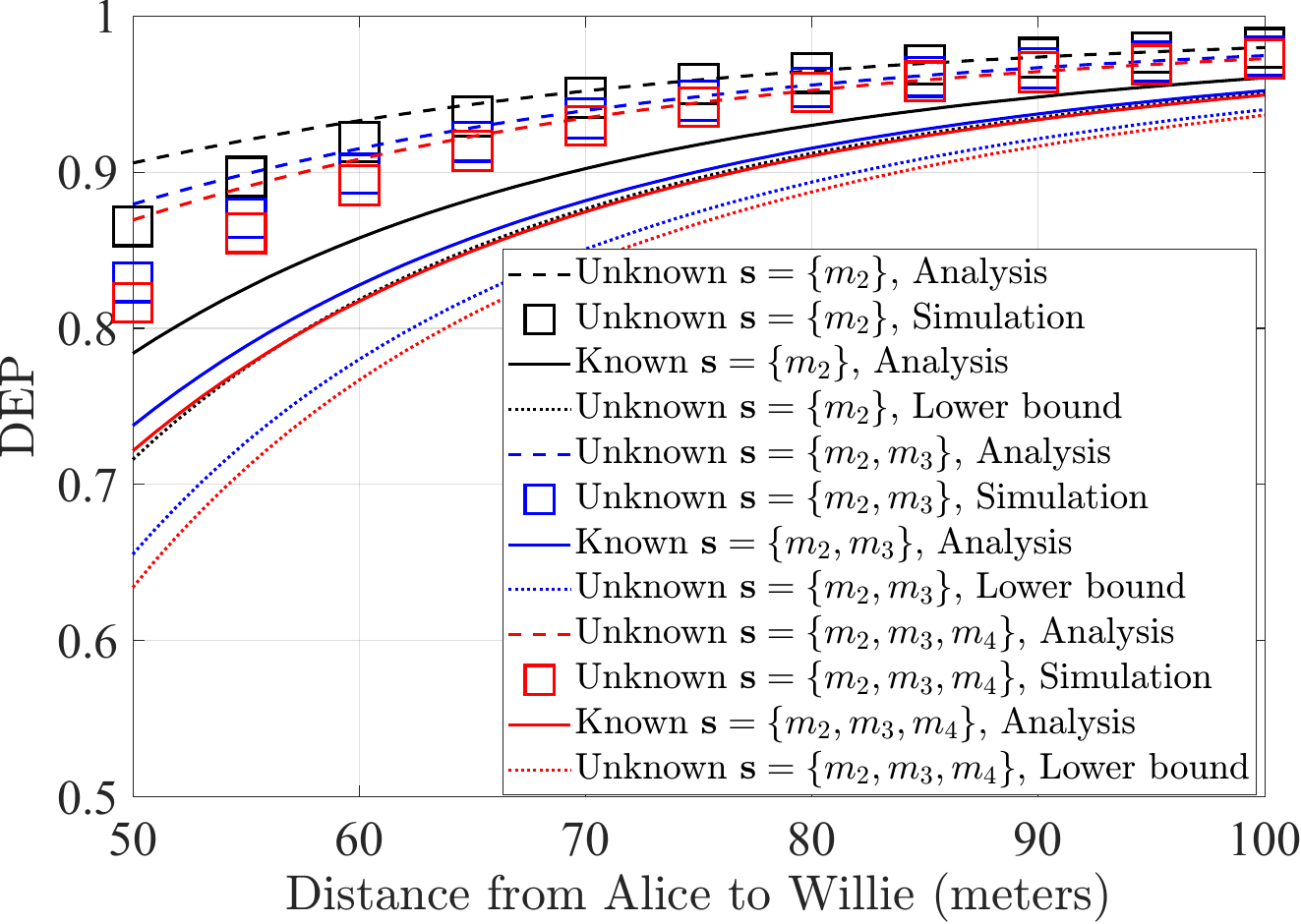}
\end{center} 
\vspace{-5.0mm}
\caption{Impact of modality uncertainty on the DEP.}
\label{figure:Unknown_1}
\end{figure}

In Fig.~\ref{figure:Unknown_1}, we compare the DEP performance under the scenario where Willie is aware of the active set of modalities with the scenario where he is unaware.
First, we observe that the DEP is significantly higher when Willie is not aware of the selected modalities, confirming that uncertainty regarding the transmission strategy improves covertness. 
Furthermore, the gap between the low-SNR regime analysis in~\eqref{eq:DEP_unknown} and the simulated exact DEP (obtained using the optimal detector in~\eqref{eq:LRT_unknown_final}) becomes negligible when the distance to Willie exceeds 70~m.
Since covert communications typically operate in the low-SNR regime, this confirms that our analysis is tight in the practical scenarios of interest. 
Finally, the proposed low-SNR analysis proves to be significantly tighter than the conventional lower bound across the entire range of distances.

\begin{figure}[t]
\begin{center} 
\includegraphics[width=3.3in]{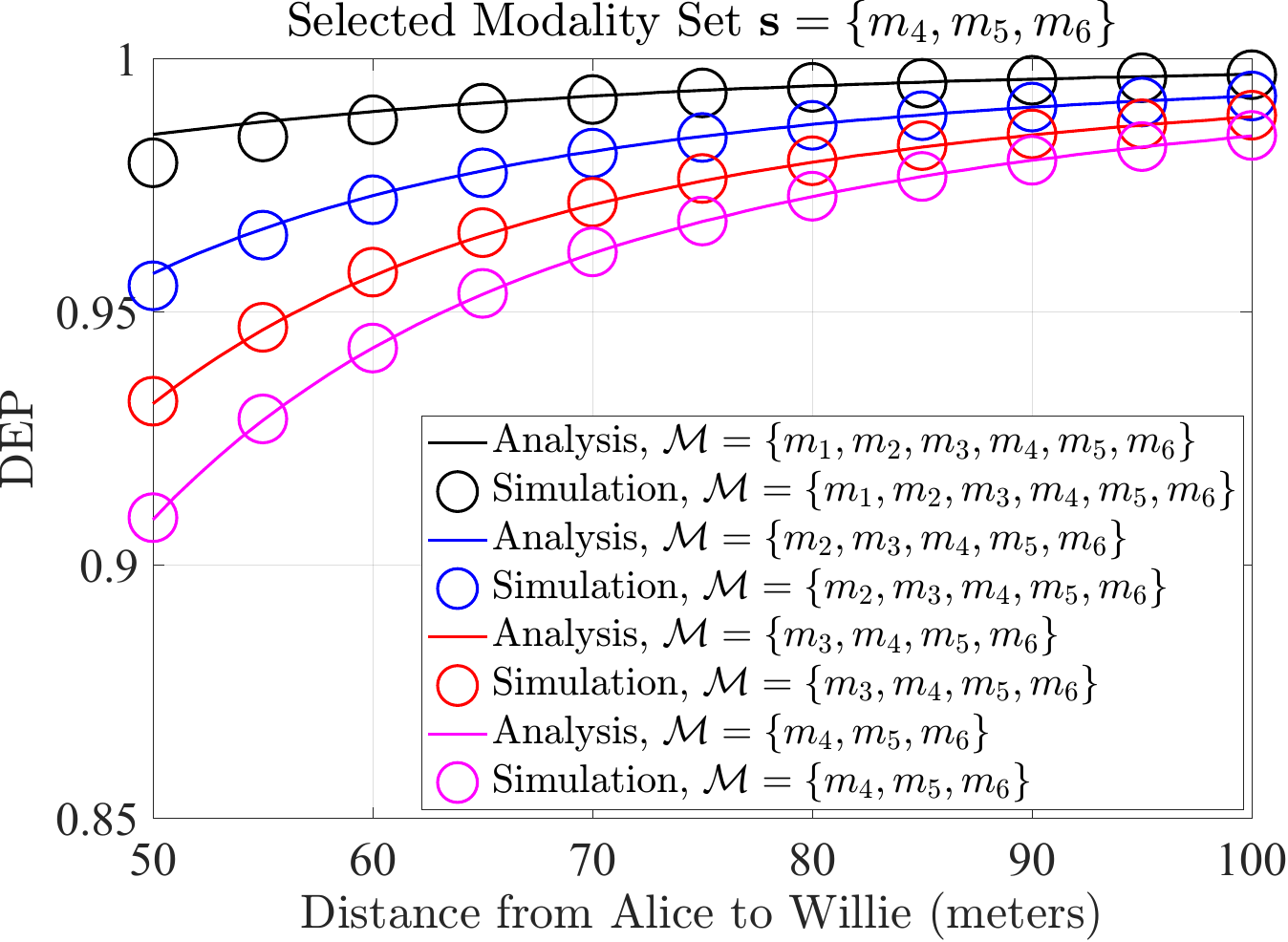}
\end{center} 
\vspace{-5.0mm}
\caption{Impact of the number of available modalities on the DEP when the selected modalities are unknown to Willie.}
\label{figure:Unknown_2}
\end{figure}

Fig.~\ref{figure:Unknown_2} illustrates the DEP performance for different sets of available modalities $\cM$, given a fixed active set $\sv = \{m_4, m_5, m_6\}$. 
First, the results confirm that the low-SNR regime analysis accurately predicts the actual DEP, although the approximation becomes slightly less tight as the number of available modalities increases.
Crucially, we observe that the DEP increases significantly when the selection is made from a larger pool of available modalities. 
For example, when $\cM= \{m_4, m_5, m_6\}$, there is no uncertainty regarding the active set, resulting in the lowest DEP. 
On the other hand, when $\cM= \{m_1,m_2,m_3,m_4, m_5, m_6\}$, Willie is forced to monitor all six frequency ranges. 
The noise collected from the inactive bands ($m_1, m_2, m_3$) degrades his detection performance, thereby leading to a higher DEP. 

\subsection{DEP performance with the modality selection optimization}

In this subsection, we evaluate the average DEP performance for $10^4$ channel realizations where the active modality set is optimized for each specific realization.
We compare the following selection strategies: 
\begin{itemize}

    \item \textit{Optimal:} This approach performs an exhaustive search over all possible non-empty subsets. It evaluates the exact DEP for every candidate subset using the optimal detectors in~\eqref{eq:decision_criteria} and~\eqref{eq:LRT_unknown_final}, and selects the one that maximizes DEP while satisfying the rate constraint. The complexity is $\mathcal{O}(2^M)$.

    \item \textit{Proposed:} The low-complexity strategy in Algorithm~\ref{Algorithm:proposed}, utilizing the efficiency metric $\Psi_m$. The complexity is bounded by $\mathcal{O}(M^2)$.
    
    \item \textit{MaxDEP-Greedy:} A baseline that calculates the DEP for each individual modality in~\eqref{eq:DEP_from_CF} or~\eqref{eq:DEP_unknown} depending on Willie's knowledge of $\sv$, sorts them in descending order, and sequentially adds them to the active set $\sv$ until the rate requirement is met. The complexity is $\mathcal{O}(M \log M)$.

    \item \textit{Random-Selection:} A naive baseline that randomly selects a modality and sequentially adds distinct random modalities to $\sv$ until the rate requirement is satisfied. The complexity is $\mathcal{O}(M)$.

\end{itemize}

\begin{figure}[t]
\begin{center} 
\includegraphics[width=3.3in]{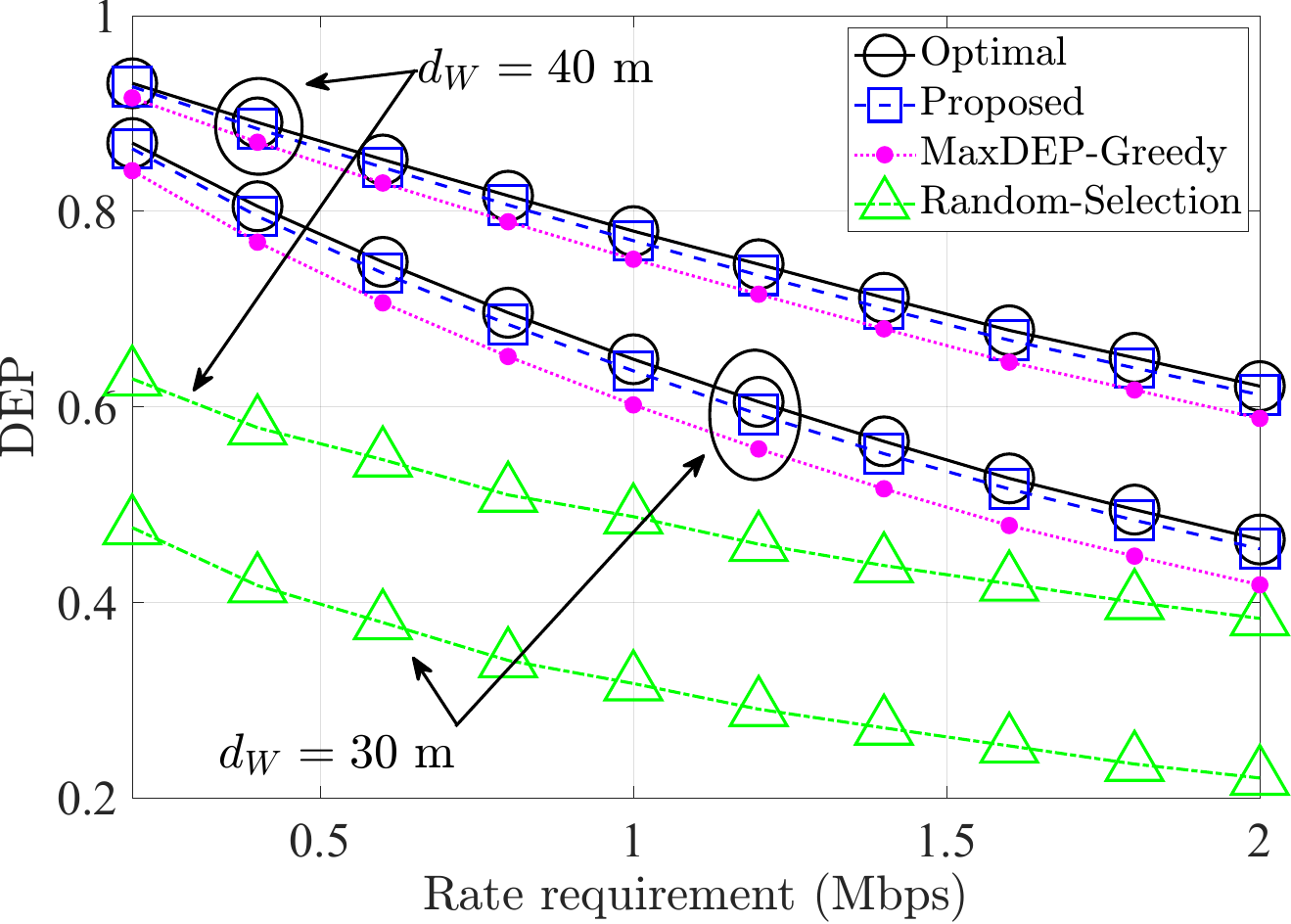}
\end{center} 
\vspace{-5.0mm}
\caption{Performance comparison of modality set selection strategies.}
\label{figure:optimization_1}
\end{figure}

In Fig.~\ref{figure:optimization_1}, we compare the DEP performance of the various modality set selection strategies as a function of the rate requirement at Bob, $U_{\text{target}}$. 
First, we observe that the proposed algorithm achieves near-optimal performance while significantly reducing computational complexity. 
Moreover, it consistently outperforms both the MaxDEP-Greedy and Random-Selection benchmarks.
We also observe that the achieved DEP decreases as the target rate increases. 
This occurs because satisfying a higher rate requirement necessitates the activation of additional modalities or those with lower center frequencies, which results in increased aggregate energy at Willie.
Lastly, as expected, the DEP increases with the distance to Willie $d_{\text{W}}$, since the received signal energy at Willie decays due to path loss.

Fig.~\ref{figure:optimization_2} demonstrates the impact of Alice's knowledge regarding the channels to Willie on the achieved DEP when $d_{\text{W}} = 40$~m. 
We evaluate the proposed algorithm using the three efficiency metrics derived in~\eqref{eq:efficiency_metric}, \eqref{eq:stat_metric}, and~\eqref{eq:metric_no_csi} for the scenarios with instantaneous CSI, statistical CSI, and no CSI, respectively, assuming small-scale fading channels follow a Nakagami fading distribution with parameter $\kappa_m=2$. 
As expected, the achieved DEP degrades as the quality of channel information available to Alice diminishes. 
However, it is worth noting that the proposed technique with only statistical CSI achieves performance comparable to the MaxDEP-Greedy scheme, which relies on full instantaneous CSI.
Furthermore, in the absence of any channel knowledge (no CSI), the proposed strategy significantly outperforms the Random-Selection baseline, demonstrating the efficiency of the proposed metric.

\begin{figure}[t]
\begin{center} 
\includegraphics[width=3.3in]{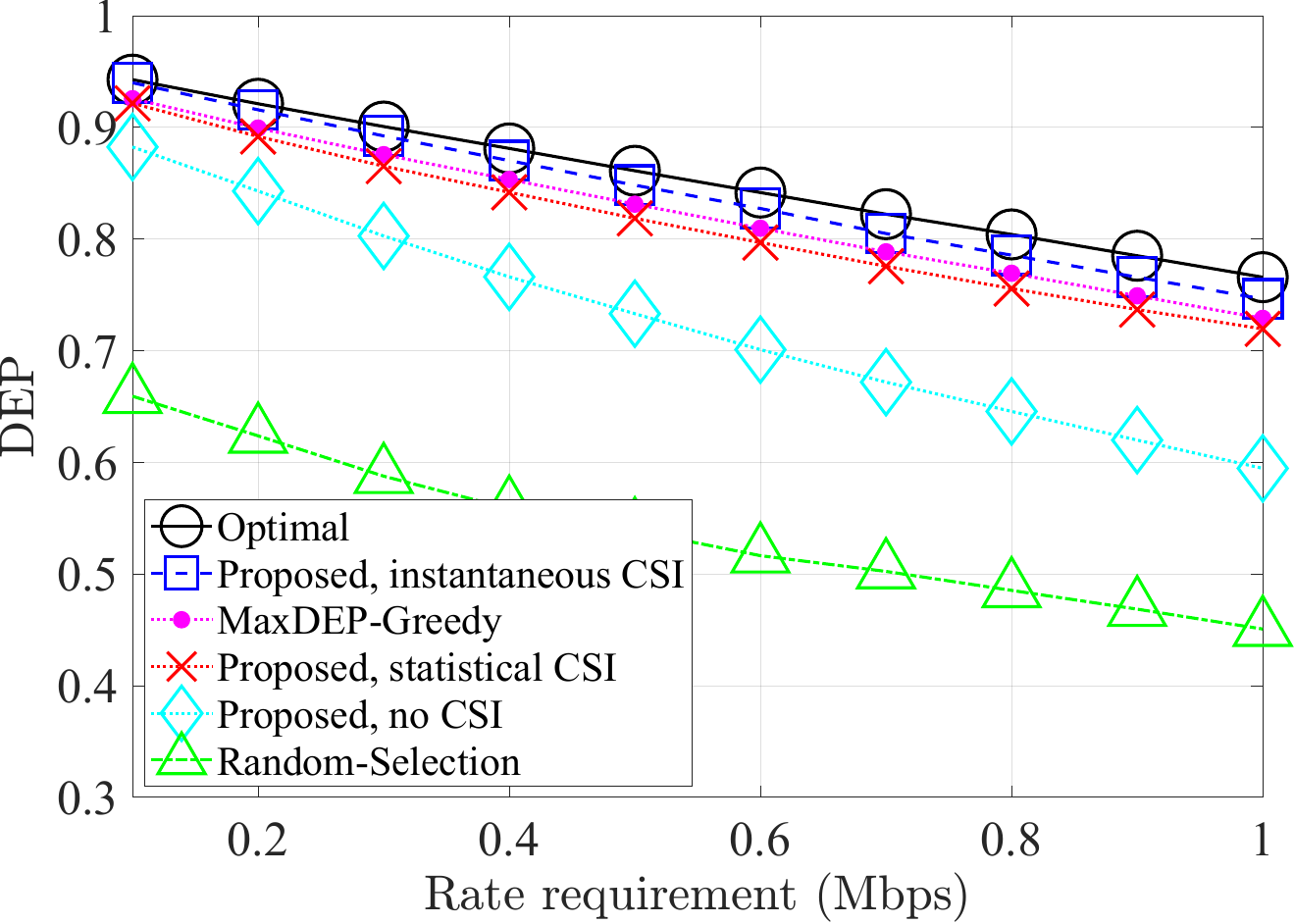}
\end{center} 
\vspace{-5.0mm}
\caption{Impact of CSI availability at Alice on the optimized DEP.}
\label{figure:optimization_2}
\end{figure}

\begin{figure}[t]
\begin{center} 
\includegraphics[width=3.3in]{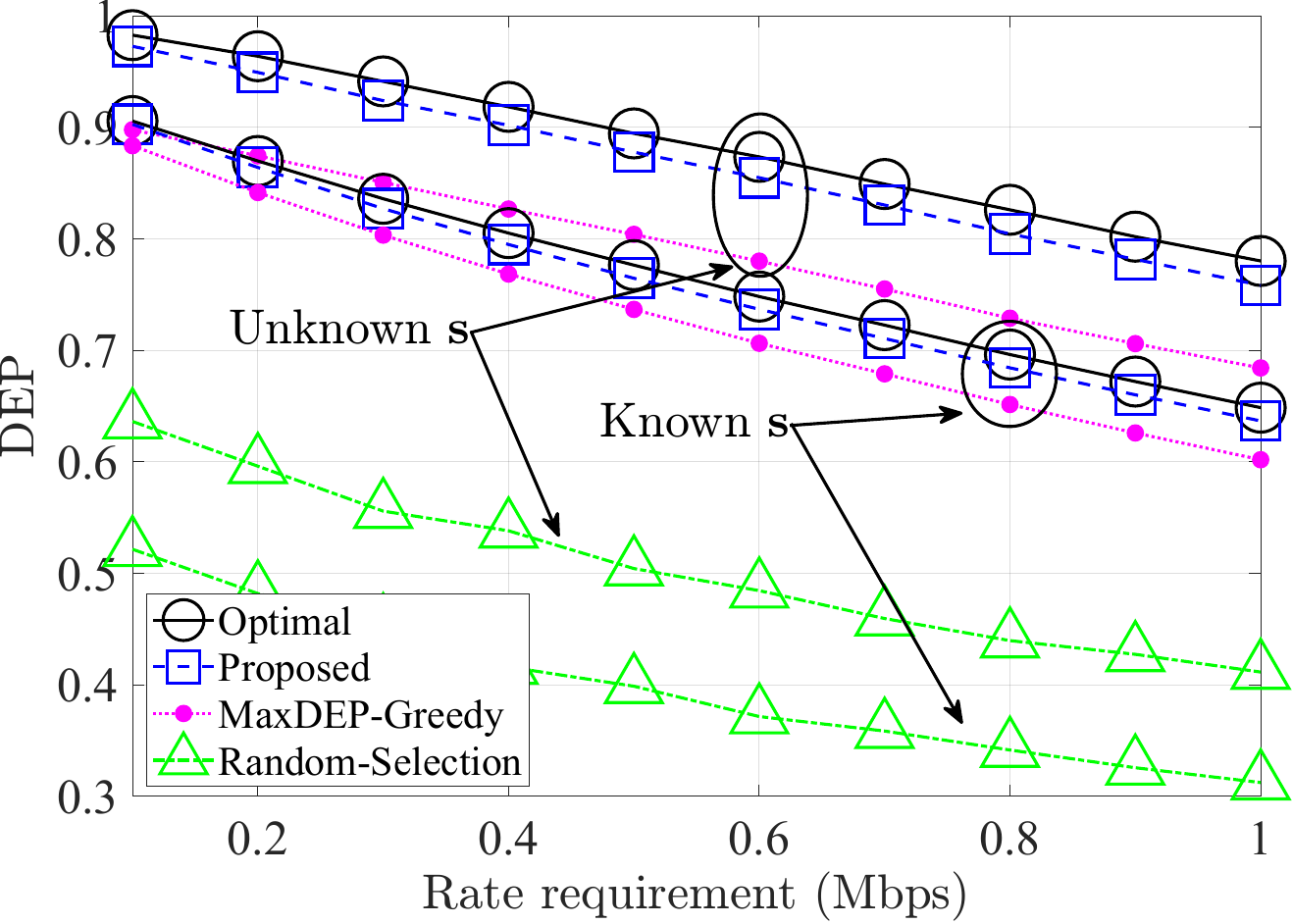}
\end{center} 
\vspace{-5.0mm}
\caption{Optimized covertness: Informed vs. uninformed Willie.}
\label{figure:optimization_3}
\end{figure}

In Fig.~\ref{figure:optimization_3}, we illustrate the optimized DEP performance, comparing the scenarios where Willie possesses knowledge of the active modality set $\sv$ versus the case where he does not, when $d_{\text{W}} = 30$~m.
It is observed that the DEP is significantly improved when $\sv$ is unknown to Willie, confirming that uncertainty regarding the transmission strategy enhances covertness. 
Furthermore, the proposed algorithm achieves near-optimal performance in both scenarios (known and unknown $\sv$), closely tracking the exhaustive search.
In addition, we observe that the performance gap between the proposed approach and the MaxDEP-Greedy scheme is more pronounced when $\sv$ is unknown to Willie. 
This is because the MaxDEP-Greedy method, by ignoring rate efficiency, tends to activate a larger number of modalities to meet the requirement, resulting in excessive aggregate energy at Willie. 
Consequently, this excess energy is particularly detrimental to the achievable DEP when Willie monitors the entire spectrum.

\section{Conclusion} \label{sec:conclusion}

In this paper, we investigated simultaneous multi-modal covert communications in heterogeneous networks, where a subset of available modalities is chosen for data transmission.
First, for the scenario where the selected modalities are known to Willie, we identified the optimal detector and derived both the exact DEP expression and a tight approximation.
Subsequently, for the case where Willie is uncertain regarding the active modalities, we derived the optimal detector and analyzed the DEP in the low-SNR regime.
Building on this analysis, we developed a novel low-complexity modality set selection algorithm that maximizes the DEP while satisfying a rate requirement at Bob.
Finally, the accuracy of our analysis and the efficiency of the proposed technique were validated through extensive numerical simulations.
Possible future research directions include extending this framework to multi-hop and multi-flow networks, as well as the joint optimization of transmit power and modality selection. 
Another interesting future research direction could involve multiple adversaries with different levels of coordination.


\appendices

\section{Proof of Theorem~\ref{Theorem:optimal_detector_known_set} } \label{appendix:optimal_detector_known_set}

The log-likelihood ratio (LLR) is given by
\begin{align} \label{eq:LLR}
	&\ln\left( \Lambda(Y_{\text{W}}) \right) = \ln\left( \frac{f(Y_{\text{W}} |\cH_1)}{f(Y_{\text{W}} |\cH_0)} \right)  \\
                   &= \ln\left(f(Y_{\text{W}} |\cH_1)\right) - \ln\left( f(Y_{\text{W}} |\cH_0) \right) \nonumber \\
                   &= \sum_{m \in \sv} \sum_{l=1}^L \Bigg[  \ln\left(  \frac{1}{ \pi  \sigma_{1,m}^2 } \right) - \frac{ \left| y_{\text{W},m}[l] \right|^2 }{   \sigma_{1,m}^2  }  \nonumber \\
                   &\qquad \qquad \quad  ~- \ln\left(  \frac{1}{ \pi \sigma_{0,m}^2 } \right)  + \frac{ \left| y_{\text{W},m}[l] \right|^2 }{  \sigma_{0,m}^2  }      \Bigg] \nonumber \\
                   &= \sum_{m \in \sv} \sum_{l=1}^L \left[  \left( \frac{1}{\sigma_{0,m}^2} - \frac{1}{ \sigma_{1,m}^2} \right)     \left| y_{\text{W},m}[l] \right|^2    \right] \nonumber \\
                   &\quad -  \sum_{m \in \sv} \sum_{l=1}^L \ln\left(  \frac{ \sigma_{1,m}^2}{\sigma_{0,m}^2}     \right) \nonumber \\
                   &= \sum_{m \in \sv} \left( \frac{1}{\sigma_{0,m}^2} \!-\! \frac{1}{ \sigma_{1,m}^2} \right)  \sum_{l=1}^L  \left| y_{\text{W},m}[l] \right|^2 - L \sum_{m \in \sv}  \ln\left(  \frac{ \sigma_{1,m}^2}{\sigma_{0,m}^2}     \right). \nonumber 
\end{align}
Here, we define the weight $w_m$ and the detection threshold $\delta_{\sv}$ as
\begin{align} \label{eq:weight_derivation}
    w_m &\triangleq  \frac{1}{\sigma_{0,m}^2} - \frac{1}{ \sigma_{1,m}^2} =\frac{\rho_{\text{W},m}}{ \Omega_m N_0 \left( 1+\rho_{\text{W},m} \right)   }, \\ \label{eq:threshold_derivation}
    \delta_{\sv} &\triangleq L \sum_{m \in \sv} \ln\left(  \frac{ \sigma_{1,m}^2}{\sigma_{0,m}^2}     \right) = L \sum_{m \in \sv} \ln\left(  1 + \rho_{\text{W},m}     \right),
\end{align}
respectively, where $\rho_{\text{W},m}$ denotes the SNR at Willie for modality $m$ as in~\eqref{eq:SNR_W_m}. 
Finally, based on the LRT in~\eqref{eq:LRT} and LLR in~\eqref{eq:LLR}, the optimal decision rule at Willie becomes
\begin{align} \label{eq:decision_criteria_appx}
	\sum_{m \in \sv} w_m E_m \underset{\cD_0}{\overset{\cD_1}{\gtrless}} \delta_{\sv},
\end{align} 
where $E_m$ is defined in~\eqref{eq:def_Energy}.

\section{Proof of Theorem~\ref{Theorem:exact_DEP} } \label{appendix:exact_DEP}

Note that $y_{\text{W},m}[l] \sim \cC\cN (0, \sigma_{k,m}^2 )$ for $\cH_k$ as in~\eqref{eq:Received_signal_Willie_distribution}.
Thus, the received signal energy $E_m =   \sum_{l=1}^L  \left| y_{\text{W},m}[l] \right|^2$ follows a chi-square distribution with degrees of freedom $2L$, i.e., $E_m \sim \frac{{\sigma_{k,m}^2}}{2} \chi^2(2L)$.
For a scalar $c>0$ and a random variable $X \sim \chi^2(n)$, $cX$ follows the gamma distribution with shape $\frac{n}{2}$ and scale $2c$, i.e., $cX \sim \cG(\frac{n}{2},2c)$.
Hence, for $m \in \sv$, we have 
\begin{align} \label{eq:Energy_each_mod_distribution}
	&\cH_k:  	w_m E_m \sim   \cG\left(L, w_m \sigma_{k,m}^2 \right),
\end{align}
for $k \in \{0, 1\}$.

The CF of the test statistic $T_{\sv} = \sum_{m \in \sv} w_m E_m$ is derived as
\begin{align} \label{eq:CF_of_T_appendix}
	\Phi_{T_{\sv}|\cH_k}(t) &\triangleq \EE\left[ e^{it\sum_{m \in \sv} w_m E_m}\right]  = \prod_{m \in \sv} \EE\left[ e^{it w_m E_m}\right],
\end{align}
where the CF of the gamma-distributed $w_m E_m$ is
\begin{align}   \label{eq:phi_def}
    \phi_{k,m}(t) &\triangleq \EE\left[ e^{it w_m E_m}\right] = \left( \frac{1}{1 - i t w_m \sigma_{k,m}^2} \right)^{L}.
\end{align}

From the Gil-Pelaez inversion formula~\cite{GIL-PELAEZ:51}, the cumulative distribution function (CDF) of $X$ with the CF $\Phi_X(t)$ can be calculated as $\PP(X\leq x) = \frac{1}{2} - \frac{1}{\pi} \int_0^\infty \frac{1}{t} \Im \left[ e^{-itx}  \Phi_X(t) \right] \d t$ where $\Im[x]$ denotes the imaginary part of $x$. 
Then, we obtain the CDF of $T_{\sv}$ as 
\begin{align} \label{eq:CDF_of_T_from_CF}
    F_{T_{\sv} | \cH_k}(x) &= \PP( T_{\sv} \leq x | \cH_k ) \nonumber \\
    &= \frac{1}{2} - \frac{1}{\pi} \int_0^\infty \frac{1}{t} \Im \left[ e^{-itx}  \Phi_{T_{\sv} | \cH_k}(t) \right] \d t.
\end{align}
By substituting~\eqref{eq:CDF_of_T_from_CF} into~\eqref{eq:DEP_def}, we identify the exact DEP expression as in~\eqref{eq:DEP_from_CF}.

\section{Proof of Lemma~\ref{Lemma:approx_DEP} } \label{appendix:approx_DEP}

Recall that the weighted energy $w_m E_m$ follows a gamma distribution $w_m E_m \sim  \cG(L,w_m \sigma_{k,m}^2 )$.
Thus, the mean and variance of the test statistic $T_{\sv} = \sum_{m \in \sv} w_m E_m$ are $\mu_{\sv,k} = L \sum_{m \in \sv} w_m \sigma_{k,m}^2$ and $\sigma_{\sv,k}^2 = L \sum_{m \in \sv} ( w_m \sigma_{k,m}^2 )^2$ for $k \in \{0, 1\}$, respectively.

Then, by applying the moment matching-based approximation~\cite{Akhiezer:65}, the approximated $T_{\sv}$ becomes
\begin{align}
    T_{\sv,\text{approx}} \sim \cG\left(  \alpha_{\sv,k}   ,  \theta_{\sv,k} \right),
\end{align}
where the shape parameter $\alpha_{\sv,k}$ and scale parameter $\theta_{\sv,k}$ are written as
\begin{align} \label{eq:shape_approx}
    \alpha_{\sv,k} &\triangleq \frac{\mu_{\sv ,k}^2}{\sigma_{\sv,k}^2} =  \frac{L \left(\sum_{m \in \sv} w_m \sigma_{k,m}^2\right)^2}{\sum_{m \in \sv} \left( w_m \sigma_{k,m}^2 \right)^2},  \\ \label{eq:scale_approx}
    \theta_{\sv,k} &\triangleq \frac{\sigma_{\sv,k}^2}{\mu_{\sv ,k}} = \frac{\sum_{m \in \sv} \left( w_m \sigma_{k,m}^2 \right)^2}{\sum_{m \in \sv} w_m \sigma_{k,m}^2}.
\end{align}

Lastly, we have an approximation of the DEP as
\begin{align} \label{eq:DEP_moment_matching_appendix}
    \tP_{\text{DEP},\sv} &\approx 1 + \PP\Big(  T_{\sv,\text{approx}} < \delta_{\sv} \Big| \cH_1 \Big) \!-\! \PP\Big(  T_{\sv,\text{approx}} < \delta_{\sv} \Big| \cH_0 \Big) \nonumber \\
    &= 1 + \frac{\gamma\left( \alpha_{\sv,1}, \frac{\delta_{\sv}}{\theta_{\sv,1}} \right)}{\Gamma\left(\alpha_{\sv,1}\right)} - \frac{\gamma\left( \alpha_{\sv,0}, \frac{\delta_{\sv}}{\theta_{\sv,0}} \right)}{\Gamma\left(\alpha_{\sv,0}\right)}.
\end{align}
This concludes the proof.

\section{Proof of Theorem~\ref{Theorem:optimal_detector_unknown} } \label{appendix:optimal_detector_unknown}

In this case, Willie has to examine the signals across all modalities $m \in \cM$, and thus the likelihood for the null hypothesis $\mathcal{H}_0$ is written as
\begin{align} \label{eq:Likelihood_Null_unknown}
    f(Y_{\text{W}} |\cH_0) =  \prod_{ m \in \cM} \prod_{l=1}^L \frac{1}{ \pi \sigma_{0,m}^2 } \exp\left( - \frac{ \left| y_{\text{W},m}[l] \right|^2 }{  \sigma_{0,m}^2 } \right).
\end{align}
In addition, as Willie does not know the specific subset of modalities $s \subseteq \mathcal{M}$ that Alice is utilizing, the likelihood for the alternate hypothesis $\mathcal{H}_1$ is given by
\begin{align} \label{eq:Likelihood_Alternate_unknown}
    &f(Y_{\text{W}} | \mathcal{H}_1) = \!\!\!\!\!\!\!\! \sum_{\vv \subseteq \mathcal{M}, \vv \neq \emptyset} \!\! \! \eta_{\vv} \Bigg[ \Bigg( \! \prod_{m \in \vv} \prod_{l=1}^L  \frac{1}{ \pi  \sigma_{1,m}^2 } \exp\bigg( \! - \frac{ \left| y_{\text{W},m}[l] \right|^2 }{   \sigma_{1,m}^2 } \bigg) \!\! \Bigg)  \nonumber \\ 
    &\times \Bigg( \prod_{m \notin \vv} \prod_{l=1}^L  \frac{1}{ \pi \sigma_{0,m}^2 } \exp\left( - \frac{ \left| y_{\text{W},m}[l] \right|^2 }{  \sigma_{0,m}^2 } \right) \Bigg) \Bigg],
\end{align}
where $\eta_{\vv}$ is the probability that a set $\vv \subseteq \cM$ is chosen by Alice. 
In the absence of prior knowledge regarding $\{ \eta_{\vv} \}$, Willie adopts a uniform distribution, setting $\eta_{\vv} = \eta \triangleq 1 / (2^M-1)$ for all valid subsets.

Consequently, the likelihood ratio is derived as
\begin{align} \label{eq:LLR_unknown}
    &\Lambda(Y_{\text{W}}) = \eta \sum_{ \vv \subseteq \mathcal{M}, \vv \neq \emptyset}  \prod_{m \in \vv} \prod_{l=1}^L  \frac{ \frac{1}{ \pi  \sigma_{1,m}^2 } \exp\left( - \frac{ \left| y_{\text{W},m}[l] \right|^2 }{   \sigma_{1,m}^2 } \right) }{ \frac{1}{ \pi \sigma_{0,m}^2 } \exp\left( - \frac{ \left| y_{\text{W},m}[l] \right|^2 }{  \sigma_{0,m}^2 } \right) } \nonumber \\
    &= \eta \sum_{ \vv \subseteq \mathcal{M}, \vv \neq \emptyset}  \prod_{m \in \vv}  \prod_{l=1}^L  \exp\left( w_m  \left| y_{\text{W},m}[l] \right|^2  - \ln\left(  1 + \rho_{\text{W},m} \right)  \right) \nonumber \\ 
    &= \eta \sum_{ \vv \subseteq \mathcal{M}, \vv \neq \emptyset} \prod_{m \in \vv} \exp\left( w_mE_m - L \ln \left(1 + \rho_{\text{W},m} \right) \right) \nonumber \\
    &= \eta \sum_{ \vv \subseteq \mathcal{M}, \vv \neq \emptyset} \!  \exp \! \bigg( \sum_{m \in \vv} \left( w_mE_m \!-\! L \ln \left(1 + \rho_{\text{W},m} \right) \right) \bigg),
\end{align}
where $w_m$, $E_m$, and $\rho_{\text{W},m}$ are defined in~\eqref{eq:def_weight}, \eqref{eq:def_Energy}, and~\eqref{eq:SNR_W_m}, respectively. 
Finally, from the LRT in~\eqref{eq:LRT}, we can obtain the decision rule for the optimal detector at Willie as in~\eqref{eq:LRT_unknown_final}.

\section{Proof of Theorem~\ref{Theorem:lowSNR_DEP} } \label{appendix:lowSNR_DEP}

When $|\rho_{\text{W},m}| \ll 1$, both $w_m = \frac{\rho_{\text{W},m}}{ \Omega_m N_0 (1 + \rho_{\text{W},m})}$ and $\delta_{\vv}$ are near 0, and hence $T_{\vv} - \delta_{\vv}$ in~\eqref{eq:LRT_unknown_final} becomes close to 0. 
Applying the Taylor series approximation $\exp(x) \approx 1 + x$ for $|x| \ll 1$, the test statistic $T$ in~\eqref{eq:LRT_unknown_final} can be approximated as
\begin{align} \label{eq:T_approx}
    T &\approx \eta \sum_{ \vv \subseteq \mathcal{M}, \vv \neq \emptyset} \left( 1 + T_{\vv} - \delta_{\vv} \right) \nonumber \\
    &= 1 +  \eta \sum_{ \vv \subseteq \mathcal{M}, \vv \neq \emptyset} \sum_{m \in \vv} \left( w_mE_m  -  L \ln\left(1 + \rho_{\text{W},m} \right) \right) \nonumber \\
    &= 1 +  2^{M-1}  \eta \sum_{m \in \cM} \left(  w_mE_m - L \ln \left(1 + \rho_{\text{W},m} \right) \right),  
\end{align}
where the last equality follows from the fact that the summations cover every possible non-empty subset $\vv \subseteq \mathcal{M}$ and each modality $m$ is included in $2^{M-1}$ subsets.

Consequently, the decision rule in~\eqref{eq:LRT_unknown_final} becomes
\begin{align} \label{eq:LRT_unknown_approx}
    \underbrace{\sum_{m \in \cM} w_mE_m}_{ \triangleq \tilde{T} }  \underset{\mathcal{D}_0}{\overset{\mathcal{D}_1}{\gtrless}} \underbrace{ L \sum_{m \in \cM} \ln \left(1 + \rho_{\text{W},m} \right) }_{  \triangleq  \tilde{\delta} }.
\end{align}
Note that the detector in~\eqref{eq:LRT_unknown_approx} examines all modalities since Willie does not have the knowledge about the selected subset of modalities $\sv$ unlike the case in~\eqref{eq:decision_criteria}.

The CFs of $\tilde{T}$ under the two hypotheses are given by~\eqref{eq:CF_of_T_H0} and~\eqref{eq:CF_of_T_H1}, respectively. 
Finally, from the approximated test statistic $\tilde{T}$ in~\eqref{eq:LRT_unknown_approx} and the Gil-Pelaez inversion formula~\cite{GIL-PELAEZ:51}, the DEP in the low-SNR regime is written as
\begin{align} \label{eq:DEP_unknown_appendix}
    &\tP_{\text{DEP},\sv} \approx  1 + \PP\left( \tilde{T} < \tilde{\delta} \Big| \cH_1 \right) - \PP\left(  \tilde{T} < \tilde{\delta} \Big| \cH_0 \right) \nonumber \\
    &= 1 - \frac{1}{\pi} \int_0^\infty \frac{1}{t} \Im \left[ e^{-it \tilde{\delta}} \left(  \Phi_{\tilde{T} | \cH_1 }(t) - \Phi_{\tilde{T} | \cH_0 }(t) \right) \right] \d t.
\end{align}
This concludes the proof.

\section{Proof of Lemma~\ref{Lemma:approx_DEP_unknown} } \label{appendix:approx_DEP_unknown}

Under the null hypothesis $\cH_0$, the mean and variance of $\tilde{T} = \sum_{m \in \cM} w_mE_m$ are respectively given by
\begin{align} \label{eq:mean_null}
    \tilde{\mu}_0 &= L \sum_{m \in \cM} w_m \sigma_{0,m}^2, \\ \label{eq:variance_null}
    \tilde{\sigma}_0^2 &= L \sum_{m \in \cM}  w_m^2 \sigma_{0,m}^4.
\end{align}
Also, under the alternate hypothesis $\cH_1$, the mean of $\tilde{T}$ is expressed as
\begin{align} \label{eq:mean_alternate}
    \tilde{\mu}_1 &= L \left(   \sum_{m \in \sv} w_m \sigma_{1,m}^2 + \sum_{m \notin \sv} w_m \sigma_{0,m}^2     \right).
\end{align}
As the noise is independent across modalities, the variance of $\tilde{T}$ under $\cH_1$ is given by
\begin{align} \label{eq:variance_alternate}
    \tilde{\sigma}_1^2 &= \sum_{m \in \sv} \text{Var}\left(  w_mE_m \big| \cH_1 \right) + \sum_{m \notin \sv} \text{Var}\left(  w_mE_m \big| \cH_0 \right) \nonumber \\
    &= L \left( \sum_{m \in \sv}  w_m^2 \sigma_{1,m}^4 + \sum_{m \notin \sv}  w_m^2 \sigma_{0,m}^4 \right).
\end{align}

Then, by adopting the two-moment matching method~\cite{Akhiezer:65}, the approximation of
$\tilde{T}$ becomes
\begin{align}
    \tilde{T}_{\text{approx}} \sim \cG\left(  \tilde{\alpha}_{k}   ,  \tilde{\theta}_{k} \right),
\end{align}
where $\tilde{\alpha}_{k} \triangleq \frac{\tilde{\mu}_{k}^2}{\tilde{\sigma}_{k}^2}$ and $\tilde{\theta}_{k} \triangleq \frac{\tilde{\sigma}_{k}^2}{\tilde{\mu}_{k}}$. 
Finally, we have an approximation of the DEP as
\begin{align} \label{eq:DEP_moment_matching_appendix}
    \tP_{\text{DEP},\sv} &\approx 1 + \PP\left( \tilde{T}_{\text{approx}} < \tilde{\delta} \Big| \cH_1 \right) - \PP\left(  \tilde{T}_{\text{approx}} < \tilde{\delta} \Big| \cH_0 \right) \nonumber \\
    &= 1 + \frac{\gamma\left( \tilde{\alpha}_{1}, \frac{\tilde{\delta}}{ \tilde{\theta}_{1}} \right)}{\Gamma\left(\tilde{\alpha}_{1}\right)} - \frac{\gamma\left( \tilde{\alpha}_{0}, \frac{ \tilde{\delta}}{\tilde{\theta}_{0}} \right)}{\Gamma\left(\tilde{\alpha}_{0}\right)}.
\end{align}
This concludes the proof.

\bibliographystyle{ieeetr}
\bibliography{bibliography.bib}

\end{document}